%% file: ecc-tr.tex
\documentclass[11pt]{article}

\usepackage{fullpage,verbatim,amsmath,amssymb,amsthm,float,color}
\usepackage[utf8]{inputenc}
\usepackage{setspace}

\ifx\pdftexversion\undefined
  \usepackage[dvips]{graphicx}
\else
  \usepackage[pdftex]{graphicx}
\fi

\newtheorem{lemma}{Lemma}

\newtheorem{theorem}{Theorem}

\floatstyle{ruled}
\newfloat{algorithm}{thp}{lop}
\floatname{algorithm}{Algorithm}

\floatstyle{ruled}
\newfloat{algorithm}{thp}{lop}
\floatname{algorithm}{Algorithm}

\newcommand{\R}{\ensuremath{\mathcal{R}}}
\newcommand{\D}{\ensuremath{\mathcal{D}}}
\newcommand{\A}{\ensuremath{\mathcal{A}}}
\newcommand{\E}{\ensuremath{\mathcal{E}}}

\newcommand{\algA}{\ensuremath{\mathcal{A}}}
\newcommand{\correct}{\mathit{correct}}
\newcommand{\faulty}{\mathit{faulty}}

\newcommand{\C}{\ensuremath{\mathsf{C}}}
\newcommand{\EC}{\ensuremath{\mathsf{EC}}}
\newcommand{\EIC}{\ensuremath{\mathsf{EIC}}}
\newcommand{\ETOB}{\ensuremath{\mathsf{ETOB}}}
\newcommand{\TOB}{\ensuremath{\mathsf{TOB}}}
\newcommand{\ignore}[1]{}

\def\Time{\mathbb{T}}

\def\fd{failure detector}

\def\to{\rightarrow}

\def\Nat{\ensuremath{\mathbb{N}}}

\def\get{\leftarrow}

\newcommand{\id}[1]{\mbox{\textit{#1}}}% for identifiers in code
\newcommand{\res}[1]{\mbox{\textbf{#1}}}% reserved words

\ignore{
\ifx\pdftexversion\undefined
  \usepackage[dvips]{graphicx}
\else
  \usepackage[pdftex]{graphicx}
  \DeclareGraphicsRule{*}{mps}{*}{}
\fi
}

\input{macros}

\begin{document}
\sloppy

\title{The Weakest Failure Detector for Eventual Consistency}

\author{
Swan Dubois\protect\footnote{Sorbonne Universit\'es, UPMC Universit\'e
  Paris 6, \'Equipe REGAL, LIP6, F-75005, Paris, France;
CNRS, UMR 7606, LIP6, F-75005, Paris, France;
Inria, \'Equipe-projet REGAL, F-75005, Paris, France;
firstname.lastname@lip6.fr}
\hspace{1cm}%\and
Rachid Guerraoui\protect\footnote{\'Ecole Polytechnique F\'ed\'erale de Lausanne, Switzerland; rachid.guerraoui@epfl.ch}
\hspace{1cm}%\and
Petr Kuznetsov\protect\footnote{T\'el\'ecom ParisTech;
   petr.kuznetsov@telecom-paristech.fr}\thanks{The research leading to these results has
  received funding from the Agence Nationale de la Recherche,  under
  grant agreement N ANR-14-CE35-0010-01, project DISCMAT.}
\\
%\and
Franck Petit\protect\footnotemark[1]
\hspace{1cm} %\and
Pierre Sens\protect\footnotemark[1]
}

\date{}

\maketitle

\begin{abstract}
In its classical form,  a \emph{consistent} replicated service requires
all replicas to witness the same evolution of the service state.
%It is known that,
Assuming a message-passing environment with a majority of
correct processes, the necessary and sufficient information about
failures for implementing a general state machine replication scheme
ensuring consistency is captured by the $\Omega$ failure detector.
%More generally, the composition of $\Omega$ an the quorum failure
%detector $\Sigma$ gives the weakest failure detector for implementing
%consistency in any environment.
%In a distributed message passing system with a majority of correct processes, the
%weakest failure detector to implement consistency in its classical
%form (e.g., state machine replication, total order broadcast,
%consensus) is $\Omega$.

This paper shows that
%[[PK do not find it surprising
%maybe surprisingly,
%]]
in such a message-passing environment,  $\Omega$ is also the weakest failure detector
to implement
%[[PK too heavy for an abstract?
%a state machine replication scheme that ensures
an \emph{eventually consistent} replicated service,
%]]
where replicas are expected to agree on the evolution of the service state only after some
(\emph{a priori} unknown) time.

In fact, we show that $\Omega$ is the weakest to implement
eventual consistency in \emph{any} message-passing environment, {\em i.e.}, under any assumption on when and
%[[PK more precise
%how many
where
%]]
failures might occur.
Ensuring (strong) consistency in any environment
requires, in addition to $\Omega$,
the quorum failure detector
%[[PK not needed?
%, {\em a.k.a},
$\Sigma$.
%]]
Our paper thus captures, for the first time,
%the exact difference,
%in terms of minimal information about failures,
an exact computational difference between
building a replicated state machine that ensures consistency and
one that only ensures eventual consistency.
\end{abstract}

%==========================================
\section{Introduction}
%==========================================

State machine replication~\cite{Lam77,Sch90}  is the most studied
technique to build a highly-available and consistent distributed
service.
Roughly speaking, the idea consists in replicating the service, modeled as a state machine,
over several processes and ensuring that all replicas behave like one
correct and available state machine,
despite concurrent invocations of operations and failures of
replicas.
This is typically captured using the abstraction of a \emph{total
  order broadcast}~\cite{CT96}, where messages represent invocations
of the service operations from clients to replicas (server
processes).
Assuming that the state machine is
deterministic, delivering the invocations in the same total order
ensures indeed that the replicas behave like a single state
machine. Total order broadcast is, in turn, typically implemented by
having the processes agree on which batch of messages to execute next,
using the \emph{consensus} abstraction~\cite{Lam98,CT96}.
(The two abstractions, consensus and total order broadcast, were shown to be equivalent~\cite{CT96}.)

Replicas behaving like a single one is a
property generally called \emph{consistency}. The
sole purpose of the  abstractions underlying the state machine
replication scheme, namely consensus and total order broadcast, is precisely to ensure
this consistency, while providing at the same time \emph{availability},
namely that the replicated service does not stop responding. The inherent costs
of these abstractions are sometimes considered too high, both in
terms of
%[[PK not sure why it is called computation
% computation,  {\em i.e.},
%]]
the necessary computability assumptions about the underlying
system~\cite{FLP85,CHT96,Bre00}, and
%[[PK for symmetry
%complexity, {\em e.g.},
%]]
the number of communication steps needed to deliver an
invocation~\cite{Lam98,Lam06}.

An appealing approach to circumvent these costs is to trade
consistency with what is sometimes called \emph{eventual consistency}~\cite{SS05, Vogels2009}:
namely to give up the requirement that the replicas  \emph{always} look the
same, and replace it with the requirement that they only look the same
%[[PK sometimes sounds confusing
%\emph{sometimes}. This very notion is typically captured by the
%concept of \emph{eventuality}.
\emph{eventually}, {\em i.e.}, after a finite but not \emph{a priori} bounded
period of time.
%]]
Basically, \emph{eventual consistency} says that the replicas can
diverge for some period, as long as this period is finite.

Many systems claim to implement general state machines that ensure
eventual consistency in message-passing  systems, {\em
  e.g.},~\cite{Cassandra, DeCandia2007}.
But, to our knowledge,  there has been no theoretical study of the
exact assumptions on the information about failures underlying those implementations.
% needed to implement this, seemingly weaker property than consistency.
This paper is the first to do so:  using the  formalism of failure detectors~\cite{CT96,CHT96},
it addresses the question of the minimal information about failures needed %in a message-passing system
to implement an eventually consistent replicated state machine.

It has been shown in~\cite{CHT96} that, in a message-passing environment with
a majority of correct processes, the weakest failure detector to
implement consensus (and, thus, total order broadcast~\cite{CT94}) is the \emph{eventual
  leader} failure detector, denoted $\Omega$. In short, $\Omega$ outputs, at
every process, a \emph{leader} process so that, eventually, the same
correct process is considered leader by all.
% Given that consensus (or total-order broadcast) is necessary and
% sufficient to implement state machine replication,
$\Omega$ can thus be viewed as the weakest failure detector to implement a
generic replicated state machine ensuring
consistency and availability in an environment with a majority of correct
processes.

%[[PK redundant?
%It is then natural to ask ``{\em what failure detector is needed in such an environment
%if only eventual consistency is required?}''
We show in this paper that, maybe
surprisingly, the weakest failure detector to implement an
\emph{eventually consistent}  replicated service in this environment
(in fact, in \emph{any} environment) is still $\Omega$.
%]]
We prove our result via an interesting generalization of the
celebrated ``CHT proof'' by Chandra, Hadzilacos and Toueg~\cite{CHT96}.
In the CHT proof, every process periodically extracts the identifier of a process that is
expected to be correct (the \emph{leader}) from the \emph{valencies} of an ever-growing collection
of locally simulated runs. 
We carefully adjust the notion of valency to apply this approach to the weaker
abstraction of \emph{eventual consensus}, which we show to be necessary
and sufficient to implement eventual consistency.

Our result becomes less surprising if we
realize that a correct majority prevents the system from being
\emph{partitioned}, and we know that both consistency and availability
cannot be achieved while tolerating partitions~\cite{Bre00,GL02,DFG10}.
Therefore, in a system with a correct majority of processes, there is no gain in
weakening consistency:  (strong) consistency requires the same
information about failures as eventual one.
% \footnote{This result is complementary to the recent study of
% \emph{eventually linearizable} computability~\cite{GR14}, where, in
% the shared-memory context, an eventually linearizable implementation of a nontrivial
% object is shown to be as hard to achieve as a linearizable one.
% Intuitively, provided a majority of correct processes we can implement
%read-write shared memory and, thus, derive the equivalence between
%consistency and eventual consistency from~\cite{GR14}.}
In an arbitrary environment, however, {\em i.e.}, under any
assumptions on when and where failures may occur,
the weakest failure detector for consistency is known to be $\Omega+\Sigma$,
where $\Sigma$~\cite{DFG10}
%[[PK not sure it is correct
%is the failure detector that basically returns a live
%\emph{quorum}~\cite{DFG10}.
returns a set of processes (called a \emph{quorum}) so that every
two such quorums intersect at any time and there is a time after which
all returned quorums contain only correct processes.
%]]
We show in this paper that ensuring eventual consistency does not
require $\Sigma$: only $\Omega$ is needed,  even if we do not assume a
 majority of correct processes. Therefore, $\Sigma$
%[[PK not sure Sigma draws a line :)
%draws a sharp line between
represents the exact difference
between consistency and eventual consistency.
%]]
Our result thus theoretically backs up partition-tolerance~\cite{Bre00,GL02} as one of the main motivations
behind the very notion of eventual consistency.

We establish our results through the following steps:
%,  each we believe is interesting in its own right:

\begin{itemize}\itemsep0pt

\item We give precise definitions of the notions of {\em eventual
  consensus} and
{\em eventual total order broadcast}.
We show that the two abstractions are equivalent. These underlie 
the intuitive notion of eventual consistency implemented in many
replicated services~\cite{DeCandia2007,Cooper2008,Chang2008}.
%RG->PK: do we really relate them to eventual consistency?

\item We show how to extend the celebrated CHT proof~\cite{CHT96}, initially establishing
that $\Omega$ is necessary for solving consensus, to the context of eventual consensus.
Through this extension, we indirectly highlight a  hidden power of the
technique proposed in~\cite{CHT96} that somehow provides more
than was used in the original CHT proof.
%[[PK not really so
%show more than what it initially claims to show.
%]]

\item We present an algorithm that uses $\Omega$ to implement, in any message-passing
  environment,  %eventual total order broadcast (and, thus,
  an eventually consistent replicated service.
%[[PK redundant?
%The algorithm can be viewed as a simple
% way to implement  a state machine
%replication scheme ensuring eventual consistency.
%]]
The algorithm features three interesting properties:
%[[PK we stated this already
%, besides not requiring $\Sigma$.
%]]
(1) An invocation can be performed after
  the optimal number of two communication steps, even if a majority of processes is not correct and even
  during periods when processes disagree on the leader, i.e., partition
  periods; \footnote{Note that three communication steps are, in the worst case, necessary when
strong consistency is required~\cite{Lam06}. }  %[[PK is this correct???]]
(2) If $\Omega$ outputs the same leader at all
processes from the very beginning, then the algorithm implements total
order broadcast and hence ensures consistency;
(3) \emph{Causal} ordering is ensured even during periods where
$\Omega$ outputs different leaders at different processes.

\end{itemize}

%[[PK
%~\cite{reference?})
%]]

The rest of the paper is organized as follows.
%\noindent
%\textbf{Roadmap.}
We present our system model and basic definitions in Section~\ref{sec:model}.
In Section~\ref{sec:defs}, we introduce abstractions for
implementing
%state machine replication that ensures
eventual consistency: namely, eventual consensus and eventual total
order broadcast, and we prove them to be equivalent.
We show in Section~\ref{sec:WFD} that the weakest failure detector for eventual consensus in any message-passing
environment is $\Omega$.
We present in Section~\ref{sec:ETOB} our algorithm that implements eventual total order broadcast using $\Omega$ in
any environment.
Section~\ref{sec:related} discusses related
work, and Section~\ref{sec:conclu} concludes the paper.
In the optional appendix, we present some proofs omitted from the main
paper, discuss an alternative (seemingly relaxed but equivalent) definition of eventual
consensus, and recall basic steps of the CHT proof.

%==========================================
\section{Preliminaries}
\label{sec:model}
%==========================================

We  adopt the classical model of distributed systems provided with the
failure detector abstraction proposed in~\cite{CT96,CHT96}.
In particular we employ the simplified version of the model proposed in~\cite{GHKT12,JT08}.

We consider a message-passing system  with a set of
	processes $\Pi=\{p_1,p_2,\ldots,p_n\}$ ($n\ge 2$).
	Processes execute steps of computation asynchronously, {\em i.e.}, there is
	no bound on the delay between steps.
However, we assume a discrete global clock
	to which the processes do not have access.
The range of this clock's ticks is $\mathbb{N}$.
Each pair of processes are connected by a reliable %(but necessarily FIFO)
link.

Processes may fail by \emph{crashing}.
%[[PK does not make sense at this point: ``halting prematurely'' what?
%, {\em i.e.},  by halting
%prematurely.
%]]
A \emph{failure pattern} is a function $F:\mathbb{N} \to 2^\Pi$,
	where $F(t)$ is the set of processes that have crashed by time $t$.
We assume that processes never recover from crashes, \emph{i. e.}, $F(t) \subseteq F(t+1)$.
Let $\faulty(F)=\bigcup_{\substack{t\in \mathbb{N}}} F(t)$
	be the set of \emph{faulty} processes in a failure pattern $F$;
	and $\correct(F)=\Pi-\faulty(F)$ be
	the set of \emph{correct} processes in $F$.
%Given a failure pattern $F$, % is clear from the context,
%	we say that process $p$ is \emph{correct} if $p\in\correct(F)$,
%	and $p$ is \emph{faulty} if $p\in\faulty(F)$.
%
An \emph{environment}, denoted $\E$, is a set of failure patterns.

%Intuitively, an environment $\E$ describes when and where failures are
%allowed to happen.

A \emph{failure detector history $H$ with range $\mathcal{R}$} is a function $H:\Pi\times\mathbb{N}\to\mathcal{R}$,
	where $H(p,t)$ is interpreted as the value output by the failure detector module
	of process $p$ at time $t$.
A \emph{failure detector $\D$ with range $\mathcal{R}$} is a function
	that maps every failure pattern $F$ to
	a nonempty set of failure detector histories.
%        describe what hints about failures and at what times processes.
$\D(F)$ denotes the set of all possible failure detector histories
	that may be output by $\D$ in a failure pattern $F$.

For example, at each process, the \emph{leader failure detector} $\Omega$ outputs the id of a process;
	furthermore, if a correct process exists, then there is a time after which
	$\Omega$ outputs the id of the same correct process at every correct process.
Another example is the \emph{quorum failure detector} $\Sigma$, which
	outputs a set of processes at each process.
Any two sets output at any times and by any processes intersect, and
	eventually every set output at any correct process
	consists of only correct processes.

An \emph{algorithm} $\A$ is modeled as a collection of $n$
deterministic automata, where $\A(p)$ specifies the behavior of
process $p$.
Computation proceeds in \emph{steps} of these automata.
In each step, identified as a tuple $(p,m,d,\algA)$, a process $p$ atomically
(1) receives a single message $m$ (that can be the empty message
$\lambda$) or accepts an \emph{input} (from the
external world),
(2)  queries its local failure detector module
	and receives a value $d$, (3) changes its state according to
        $\A(p)$, and (4) sends a message
        specified by $\A(p)$ for the new state to every process or
        produces an \emph{output} (to the external world).
Note that the use of $\lambda$ ensures that a step of a process is
always enabled, even if no message is sent to it.  	

A \emph{configuration} of an algorithm $\algA$ specifies the local
state of each process and the set of messages in transit.
In the \emph{initial} configuration of $\algA$, no message is in
transit and each process $p$ is in the initial state of the automaton $\algA(p)$.
        %(possibly encoding the input of $p$).
A \emph{schedule} $S$ of $\algA$ is a finite or infinite sequence of
steps of $\algA$ that respects $\algA(p)$  for each $p$.

Following~\cite{JT08}, we model inputs and outputs of processes using \emph{input histories}
$H_I$ and \emph{output histories} $H_O$ that specify the inputs each
process receives from its application and
the outputs each process returns to the application over time.
A \emph{run of algorithm $\algA$ using failure detector $\D$ in
  environment $\E$} is
        a tuple $R=(F,H,H_I,H_O,S,T)$, where $F$ is a failure pattern in $\E$,
        $H$ is a failure detector history in $\D(F)$,
        $H_I$ and $H_O$ are input and output histories of $\algA$,
        $S$ is a schedule of $\algA$, and
        $T$ is a list of increasing times in $\mathbb{N}$,
        where $T[i]$ is the time when step $S[i]$ is taken.
        $H\in \D(F)$, the failure detector values received by steps
        in $S$ are consistent with $H$, and $H_I$ and $H_O$ are
        consistent with $S$.
%
%A run whose schedule is finite (respectively, infinite)
%       is called a finite (respectively, infinite) run.
An infinite run of $\algA$ is \emph{admissible}
if (1) every correct process takes an infinite number of steps in $S$;
        and (2) each message sent to a correct process is eventually received.

%More precisely, for every finite prefix $S'$ of $S$,
%        and every $q\in\correct(F)$,
%        if the message buffer in configuration $S'(I)$
%        contains a message $m=(-,-,q)$,
%        then for some $i\in\N$, $S[i]=(q,m,-,\algA)$.

%\end{itemize}

We then define a distributed-computing \emph{problem}, such as consensus or total
order broadcast, as a set of tuples $(H_I,H_O)$ where $H_I$ is an
input history and $H_O$ is an output history.
% must satisfy
%in admissible runs.
An algorithm $\A$ using a failure detector $\D$ solves a problem $P$ in
an environment $\E$  if
in every admissible run of $\A$ in $\E$, the input and output histories
are in $P$.
Typically, inputs and outputs represent
invocations and responses of \emph{operations} exported by the
implemented abstraction.
If there is an algorithm that solves $P$ using $\D$, we sometimes,
with a slight language abuse, say that $\D$ \emph{implements} $P$.

Consider two problems $P$ and $P'$. A \emph{transformation from $P$ to
  $P'$ in an environment $\E$}~\cite{HT94}  is
  a map $T_{P\rightarrow P'}$ that, given any algorithm $\A_P$
  solving $P$ in $\E$, yields an algorithm $\A_{P'}$  solving $P'$ in
  $\E$.
The transformation is \emph{asynchronous} in the sense that $\A_{P}$
is used as a ``black box'' where $\A_{P'}$ is
obtained by feeding inputs to $\A_P$ and using the returned outputs to
solve $P'$.
Hence, if $P$ is solvable in $\E$ using a failure detector $\D$,
the existence of a transformation $T_{P\rightarrow P'}$ in $\E$ establishes
that $P'$ is also solvable in $\E$ using $\D$.
If, additionally, there exists a transformation from $P'$ to  $P$ in $\E$, we
say that $P$ and $P'$ are \emph{equivalent in $\E$}.

Failure detectors can be partially ordered based on their ``power'':
	failure detector $\D$ is \emph{weaker than}
	failure detector $\D'$ in $\E$ if there is an algorithm
	that \emph{emulates} the output of $\D$ using $\D'$ in $\E$~\cite{CHT96,JT08}.
If $\D$ is weaker than $\D'$, any problem that can be solved with
	$\D$ can also be solved with $\D'$.
For a problem $P$,
	$\D^*$ is the \emph{weakest} failure detector to solve $P$ in $\E$ if
	(a)~there is an algorithm that uses $\D^*$ to solve $P$ in $\E$, and
	(b)~$\D^*$ is weaker than any failure detector $\D$ that can
be used to solve $P$ in $\E$.

%==========================================
\section{Abstractions for Eventual Consistency}
\label{sec:defs}
%==========================================

We define two basic abstractions that capture the notion of eventual
consistency: eventual total order broadcast and eventual consensus.
We show that the two abstractions are equivalent: each of them can be
used to implement the other.

%\subsection{Eventual Total Order Broadcast (\ETOB)}
\paragraph{Eventual Total Order Broadcast (\ETOB)}
% PS: Remove $b_i$ from definition. The variable is never used
The \emph{total order broadcast} (\TOB) abstraction~\cite{HT94} exports one operation
$\textit{broadcastTOB(m)}$ and
maintains, at every process $p_i$, an output variable $d_i$.
Let $d_i(t)$ denote the value of $d_i$ at time $t$.
Intuitively, $d_i(t)$ is the sequence of messages $p_i$ \emph{delivered}
by time $t$.
We write $m\in d_i(t)$ if $m$ appears in $d_i(t)$.

A process $p_i$ \emph{broadcasts a message $m$ at time
$t$} by a call to $\textit{broadcastTOB(m)}$.
We say that a process $p_i$ \emph{stably delivers a message $m$ at time $t$}
if $p_i$ appends $m$ to $d_i(t)$ and $m$ is never removed from $d_i$ after
that, {\em i.e.},
$m\notin d_i(t-1)$ and $\forall t'\geq t$: $m\in d_i(t')$.
Note that if a message is delivered but not \emph{stably} delivered by $p_i$ at time $t$,
it appears in $d_i(t)$ but not in $d_i(t')$ for some $t'>t$.

Assuming that broadcast messages are distinct, the \TOB~abstraction satisfies:
\begin{description}\itemsep0pt
\item[\TOB-Validity] If a correct process $p_i$ broadcasts a message
  $m$ at time $t$, then $p_i$ eventually stably delivers $m$, i.e., $\forall t''\geq t':$ $m\in
  d_i(t'')$ for some $t'>t$.

\item[\TOB-No-creation] If $m\in d_i(t)$, then $m$ was broadcast by
  some process $p_j$ at some time $t'<t$.

\item[\TOB-No-duplication] No message appears more than once in
  $d_i(t)$. 

\item[\TOB-Agreement] If a message $m$ is stably delivered by some
  correct process $p_i$ at time $t$, then $m$ is eventually stably delivered by every correct process $p_j$.

\item[\TOB-Stability] For any correct process $p_i$, $d_i(t_1)$ is a prefix of $d_i(t_2)$ for all $t_1, t_2\in\mathbb{N}$, $t_1\leq t_2$.

\item[\TOB-Total-order] Let $p_i$ and $p_j$ be any two correct processes such that two messages $m_1$ and $m_2$ appear in $d_i(t)$ and $d_j(t)$ at time $t$. If $m_1$ appears before $m_2$ in $d_i(t)$, then $m_1$ appears before $m_2$ in $d_j(t)$.
\end{description}
We then introduce the \emph{eventual  total order broadcast} (\ETOB)
abstraction, which  maintains the same inputs and outputs as \TOB~(messages are broadcast by a call to $\textit{broadcastETOB(m)}$) and
satisfies, in every admissible run,
the \TOB-Validity, \TOB-No-creation, \TOB-No-duplication, and
\TOB-Agreement properties,
plus the following relaxed properties for some $\tau\in\mathbb{N}$:
\begin{description}\itemsep0pt
\item[\ETOB-Stability] For any correct process $p_i$, $d_i(t_1)$ is a prefix of $d_i(t_2)$
  for all $t_1, t_2\in\mathbb{N}$, $\tau\le t_1\leq t_2$.
\item[\ETOB-Total-order] Let  $p_i$ and $p_j$ be correct processes such that
  messages $m_1$ and $m_2$ appear in $d_i(t)$ and $d_j(t)$ for some
  $t\ge \tau$. If $m_1$ appears before $m_2$ in $d_i(t)$, then $m_1$ appears before $m_2$ in $d_j(t)$.
\end{description}
As we show in this paper, satisfying the following optional (but
useful) property in \ETOB~does not require more information about failures.
\begin{description}
\item[\TOB-Causal-Order] Let $p_i$ be a correct process such that two messages $m_1$ and $m_2$ appear in $d_i(t)$ at time $t\in\mathbb{N}$. If $m_2$ depends causally of $m_1$, then $m_1$ appears before $m_2$ in $d_i(t)$.
\end{description}

Here we say that a message $m_2$ \emph{causally depends} on a message $m_1$ in a run $R$,
and write $m_1\rightarrow_R m_2$, if one of the following conditions
holds in $R$: (1) a process $p_i$ sends $m_1$ and
then sends $m_2$, (2) a process $p_i$ receives $m_1$ and then sends
$m_2$, or (3) there exists $m_3$ such that $m_1\rightarrow_R m_3$ and
$m_3\rightarrow_R m_2$.

%\subsection{Eventual Consensus (\EC)}
\paragraph{Eventual Consensus (\EC)}
The \emph{consensus} abstraction (\C)~\cite{FLP85} exports, to every process $p_i$, a
single operation $\textit{proposeC}$ that takes a
binary argument and
returns a binary response (we also say \emph{decides})  so that the following properties are
satisfied:

\begin{description}\itemsep0pt
\item[\C-Termination] Every correct process eventually returns a response to $\textit{proposeC}$.
\item[\C-Integrity] Every process returns a response at most once.
\item[\C-Agreement] No two processes return different values.
\item[\C-Validity] Every value returned was previously proposed.
\end{description}
The \emph{eventual  consensus} (\EC) abstraction exports, to
every process $p_i$, operations $\textit{proposeEC}_{1}$,
$\textit{proposeEC}_{2}$, $\ldots$ that take binary arguments and return
binary responses. Assuming that, for all
$j\in\mathbb{N}$,  every process invokes
$\textit{proposeEC}_{j}$ as soon as
it returns a response to $\textit{proposeEC}_{j-1}$, the abstraction guarantees that,
in every admissible run, there exists $k\in\mathbb{N}$, such that the following properties are satisfied:
\begin{description}\itemsep0pt
\item[\EC-Termination] Every correct process eventually returns a response to $\textit{proposeEC}_{j}$ for all $j\in\mathbb{N}$.
%[[PK not sure we need this
\item[\EC-Integrity] No process responds twice to
  $\textit{proposeEC}_{j}$ for all $j\in\mathbb{N}$.
%]]
\item[\EC-Validity] Every value returned to $\textit{proposeEC}_{j}$ was previously proposed to $\textit{proposeEC}_{j}$ for all $j\in\mathbb{N}$.
\item[\EC-Agreement] No two processes return different values to $\textit{proposeEC}_{j}$ for all $j\geq k$.
\end{description}
It is straightforward to transform the binary version of {\EC} into a
multivalued one with unbounded set of inputs~\cite{MRT00-mv}. In the following, by referring to
{\EC} we mean a multivalued version of it.

%\subsection{Equivalence between \EC~and \ETOB}
\paragraph{Equivalence between \EC~and \ETOB}
It is well known that, in their classical forms, the consensus and the total order broadcast abstractions are equivalent \cite{CT96}. In this section, we show that a similar result holds for our eventual versions of these abstractions.

The intuition behind the transformation from {\EC} to {\ETOB} is the
following.  Each time a process $p_i$ wants to ETOB-broadcast a message $m$, $p$
sends $m$ to each process.
Periodically,  every process $p_i$ proposes its current
sequence of messages received so far to  \EC.
This sequence is built by concatenating the last output of  
\EC (stored in a local variable $d_i$) to the batch of all
messages received by the process and not yet present in $d_i$.
The output of {\EC} is stored in $d_i$, i.e.,
at any time, each process delivers the last sequence of messages
returned by {\EC}.

The correctness of this transformation follows from the fact that  \EC eventually returns consistent
responses to the processes. Thus, eventually, all processes agree on
the same linearly growing sequence of stably delivered messages.
Furthermore, every message broadcast by a
correct process eventually appears either in the delivered message sequence or
in the batches of not yet delivered messages at all correct
processes. Thus, by {\EC}-Validity of {\EC}, every message
\ETOB-broadcast by a correct process is eventually stored
in $d_i$ of every correct process $p_i$ forever.
By construction, no message appears in $d_i$ twice or if it was not previously
\ETOB-broadcast.
Therefore, the
transformation satisfies the properties of {\ETOB}.

The transformation from \ETOB~to \EC~is as follows. At each invocation of the
\EC~primitive, the process broadcasts a message using the
\ETOB~abstraction. This message contains the proposed value and the
index of the consensus instance. As soon as a message corresponding
to a given eventual consensus instance is delivered by process $p_i$ (appears in $d_i$), $p_i$
returns the value contained in the message.

Since the \ETOB~abstraction guarantees that every process eventually
stably delivers the same sequence of messages, there exists a consensus instance
after which the responses of the transformation to all alive processes
are identical.
Moreover, by \ETOB-Validity, every message \ETOB-broadcast by a
correct process $p_i$ is eventually stably delivered.
Thus, every correct process eventually returns from any \EC-instance
it invokes.
Thus, the transformation satisfies the \EC~specification.

\begin{theorem}\label{th:EquivalenceECETOB}
In any environment $\mathcal{E}$, \EC~and \ETOB~are equivalent.
\end{theorem}
\begin{proof}
\paragraph{From {\EC} to {\ETOB}} 
To prove this result, it is sufficient to provide a protocol that
implements \ETOB~in an environment $\mathcal{E}$ knowing that there
exists a protocol that implements \EC~in this environment. This
transformation protocol $\mathcal{T}_{\EC\rightarrow\ETOB}$ is stated
in Algorithm~\ref{algo:ECtoETOB}. 
Now,  we are going to prove that $\mathcal{T}_{\EC\rightarrow\ETOB}$ implements \ETOB.

%Validity
Assume that there exists a message $m$ broadcast by a correct process
$p_i$ at time $t$. As $p_i$ is correct, every correct process receives
the message $push(m)$ in a finite time. Then, $m$ appears in the set
$toDeliver$ of all correct processes in a finite time. Hence, by the
termination property of \EC~and the construction of the function
$NewBatch$, there exists $\ell$ such that $m$ is included in any
sequence submitted to $\textit{proposeEC}_\ell$. By the \EC-Validity
and the \EC-Termination properties, we deduce that $p_i$ stably delivers $m$ in a finite time, that proves that $\mathcal{T}_{\EC\rightarrow\ETOB}$ satisfies the \TOB-Validity property.

%No-creation
If a process $p_i$ delivers a message $m$ at time $t$, then $m$
appears in the sequence responded by its last invocation of
$\textit{proposeEC}_\ell$. By construction and by the \EC-Validity
property, this sequence contains only messages that appear in the set $toDeliver$ of a process $p_j$ at the time $p_j$ invokes $\textit{proposeEC}_\ell$. But this set is incrementally built at the reception of messages $push$ that contains only messages broadcast by a process. This implies that $\mathcal{T}_{\EC\rightarrow\ETOB}$ satisfies the \TOB-No-creation.

%No-duplication
As the sequence outputted at any time by any process is the response to its last invocation of $\textit{proposeEC}$ and that the sequence submitted to any invocation of this primitive contains no duplicated message (by definition of the function $NewBatch$), we can deduce from the \EC-Validity property that $\mathcal{T}_{\EC\rightarrow\ETOB}$ satisfies the \TOB-No-duplication.

%Agreement
Assume that a correct process $p_i$ stably delivers a message $m$,
i.e., there exists a time after which $m$ always appears in $d_i$.
By the algorithm, $m$ always appears in the response of
$\textit{proposeEC}$ to $p_i$ after this time.
As \EC-Agreement property is eventually satisfied, we can deduce that
$m$ always appears in the response of $\textit{proposeEC}$ for any
correct process after some time. Thus, any correct process stably
delivers $m$, and $\mathcal{T}_{\EC\rightarrow\ETOB}$ satisfies the \TOB-Agreement.

%Definition of tau
Let $\tau$ be the time after which the \EC~primitive satisfies \EC-Agreement and \EC-Validity.

%Stability
Let $p_i$ be a correct process and $\tau\leq t_1\leq t_2$. Let $\ell_1$ (respectively $\ell_2$) be the integer such that $d_i(t_1)$ (respectively $d_i(t_2)$) is the response of $\textit{proposeEC}_{\ell_1}$ (respectively $\textit{proposeEC}_{\ell_2}$). By construction of the protocol and the \EC-Agreement and \EC-Validity properties, we know that, after time $\tau$, the response of $\textit{proposeEC}_\ell$ to correct processes is a prefix of the response of $\textit{proposeEC}_{\ell+1}$. As we have $\ell_1\leq \ell_2$, we can deduce that $\mathcal{T}_{\EC\rightarrow\ETOB}$ satisfies the \ETOB-Stability property.

%Total-order
Let $p_i$ and $p_j$ be two correct processes such that two messages $m_1$ and $m_2$ appear in $d_i(t)$ and $d_j(t)$ at time $t\geq \tau$. Let $\ell$ be the smallest integer such that $m_1$ and $m_2$ appear in the response of $\textit{proposeEC}_\ell$. By the \EC-Agreement property, we know that the response of $\textit{proposeEC}_\ell$ is identical for all correct processes. Then , by the \ETOB-Stability property proved above, that implies that, if $m_1$ appears before $m_2$ in $d_i(t)$, then $m_1$ appears before $m_2$ in $d_j(t)$. In other words, $\mathcal{T}_{\EC\rightarrow\ETOB}$ satisfies the \ETOB-Total-order property.

In conclusion, $\mathcal{T}_{\EC\rightarrow\ETOB}$ satisfies the \ETOB~specification in an environment $\mathcal{E}$ provided that there exists a protocol that implements \EC~in this environment.

\begin{algorithm}
\caption{$\mathcal{T}_{\EC\rightarrow\ETOB}$: transformation from \EC~to \ETOB~for process $p_i$}\label{algo:ECtoETOB}
\small
\begin{description}\itemsep0pt
\item[Output variable:]~\\
$d_i$: sequence of messages of $M$ (initially empty) outputted at any time by $p_i$

\item[Internal variables:]~\\
$toDeliver_i$: set of messages of $M$ (initially empty) containing all messages received by $p_i$~\\
$count_i$: integer (initially $0$) that stores the number of the last instance of consensus invoked by $p_i$

\item[Messages:]~\\
$push(m)$  with $m$ a message of $M$

\item[Functions:]~\\
$Send(message)$ sends $message$ to all processes (including $p_i$)\\
$NewBatch(d_i,toDeliver_i)$ returns a sequence containing all messages from the set $toDeliver_i\setminus\{m|m\in d_i\}$

\item[On reception of $broadcastETOB(m)$ from the application]~\\
$Send(push(m))$

\item[On reception of $push(m)$ from $p_j$]~\\
$toDeliver_i:=toDeliver_i\cup\{m\}$

\item[On reception of $d$ as response of $\textit{proposeEC}_\ell$]~\\
$d_i:=d$\\
$count_i:=count_i+1$\\
$\textit{proposeEC}_{count_i}(d_i. NewBatch(d_i,toDeliver_i))$

\item[On local timeout]~\\
If $count_i=0$ then\\
$\left\lfloor\begin{array}{l}
count_i:=1\\
\textit{proposeEC}_{1}(NewBatch(d_i,toDeliver_i))
\end{array}\right.$

\end{description}
\normalsize
\end{algorithm}

\paragraph{From {\ETOB} to {\EC}} To prove this result, it is sufficient to provide a protocol that
implements \EC~in an environment $\mathcal{E}$ given a protocol that
implements \ETOB~in this environment. This transformation protocol
$\mathcal{T}_{\ETOB\rightarrow\EC}$ is stated  in Algorithm
\ref{algo:ETOBtoEC}. Now, we are going to prove that $\mathcal{T}_{\ETOB\rightarrow\EC}$ implements \EC.

%Termination
Let $p_i$ be a correct process that invokes $\textit{proposeEC}_{\ell}(v)$ with $\ell\in\mathbb{N}$. Then, by fairness and the \TOB-Validity property, the construction of the protocol implies that the \ETOB~primitive delivers the message $(\ell,v)$ to $p_i$ in a finite time. By the use of the local timeout, we know that $p_i$ returns from $\textit{proposeEC}_{\ell}(v)$ in a finite time, that proves that $\mathcal{T}_{\ETOB\rightarrow\EC}$ satisfies the \EC-Termination property.

%Integrity
The update of the variable $count_i$ to $\ell$ for any process $p_i$ that invokes $\textit{proposeEC}_\ell$ and the assumptions on operations $\textit{proposeEC}$ ensure us that $p_i$ executes at most once the function\linebreak $DecideEC(\ell,received_i[\Omega_i,\ell])$. Hence, $\mathcal{T}_{\ETOB\rightarrow\EC}$ satisfies the \EC-Integrity property.

%Definition of tau
Let $\tau$ be the time after which the \ETOB-Stability and the \ETOB-Total-order properties are satisfied. Let $k$ be the smallest integer such that any process that invokes $\textit{proposeEC}_k$ in run $r$ invokes it after $\tau$.

%Agreement
If we assume that there exist two correct processes $p_i$ and $p_j$ that return different values to $\textit{proposeEC}_{\ell}$ with $\ell\geq k$, we obtain a contradiction with the \ETOB-Stability, \ETOB-Total-order, or \TOB-Agreement property. Indeed, if $p_i$ returns a value after time $\tau$, that implies that this value appears in $d_i$ and then, by the \TOB-Agreement property, this value eventually appears in $d_j$. If $p_j$ returns a different value from $p_i$, that implies that this value is the first occurrence of a message associated to $\textit{proposeEC}_{\ell}$ in $d_j$ at the time of the return of $\textit{proposeEC}_{\ell}$. After that, $d_j$ cannot satisfy simultaneously the \ETOB-Stability and the \ETOB-Total-order properties. This contradiction shows that $\mathcal{T}_{\ETOB\rightarrow\EC}$ satisfies the \EC-Agreement property.

%Validity
If we assume that there exists a process $p_i$ that returns to $\textit{proposeEC}_{\ell}$ with $\ell\in\mathbb{N}$ a value that was not proposed to $\textit{proposeEC}_{\ell}$, we obtain a contradiction with the \TOB-No-creation property. Indeed, the return of $p_i$ from $\textit{proposeEC}_{\ell}$ is chosen in $d_i$ that contains the output of the \ETOB~primitive and processes broadcast only proposed values. This contradiction shows that $\mathcal{T}_{\ETOB\rightarrow\EC}$ satisfies the \EC-Validity property.

In conclusion, $\mathcal{T}_{\ETOB\rightarrow\EC}$ satisfies the \EC~specification in an environment $\mathcal{E}$ provided that there exists a protocol that implements \ETOB~in this environment.
\end{proof}

\begin{algorithm}
\caption{$\mathcal{T}_{\ETOB\rightarrow\EC}$: transformation from \ETOB~to \EC~for process $p_i$}\label{algo:ETOBtoEC}
\small
\begin{description}\itemsep0pt
\item[Internal variables:]~\\
$count_i$: integer (initially $0$) that stores the number of the last instances of consensus invoked by $p_i$\\
$d_i$: sequence of messages (initially empty) outputted to $p_i$ by the \ETOB~primitive

\item[Functions:]~\\
$First(\ell)$: returns the value $v$ such that $(\ell,v)$ is the first message of the form $(\ell,*)$ in $d_i$ if such messages exist, $\bot$ otherwise\\
$DecideEC(\ell,v)$: returns the value $v$ as response to $\textit{proposeEC}_{\ell}$

\item[On invocation of $\textit{proposeEC}_{\ell}(v)$]~\\
$count_i:=\ell$\\
$broadcastETOB((\ell,v))$

\item[On local time out]~\\
If $First(count_i)\neq\bot$ then\\
$\left\lfloor\begin{array}{l}
DecideEC(count_i,First(count_i))
\end{array}\right.$
\end{description}
\normalsize
\end{algorithm}

%
%The proof of this result lies on the existence of transformations
%between the two abstractions:
%
%\begin{lemma}\label{lem:ECtoETOB}
%In every environment $\mathcal{E}$, there exists a transformation from \EC~to \ETOB.
%\end{lemma}
%
%Due to space limitation, the technical proof of this result is moved
%to Appendix \ref{sec:ECtoETOB}.
%(The proof is delegated to Appendix~\ref{sec:EquivalenceECETOB}.)

%
%\begin{lemma}\label{lem:ETOBtoEC}
%In every environment $\mathcal{E}$, there exists a transformation from \ETOB~to \EC.
%\end{lemma}
%
%(The proof is delegated to Appendix~\ref{sec:ETOBtoEC}.)

\section{The Weakest Failure Detector for \EC}
\label{sec:WFD}

In this section, we show that $\Omega$ is necessary and sufficient for
implementing the eventual consensus abstraction \EC:

\begin{theorem}\label{th:WFDforEC}
In any environment $\mathcal{E}$, $\Omega$ is the weakest failure detector for \EC.
\end{theorem}

%\subsection{$\Omega$ is Necessary for \EC}
\paragraph{$\Omega$ is necessary for \EC}
Let $\E$ be any environment.
We show below that $\Omega$ is weaker
than any failure detector $\D$ that can be used to solve {\EC} in $\E$.
Recall that implementing $\Omega$ means outputting, at every process,
        the identifier of a \emph{leader} process so that eventually,
         the same correct leader is output permanently at all correct
        processes.

%Before presenting our algorithm extracting of $\Omega$ from any {\EC}
%algorithm using $\D$,
First,  we briefly recall the arguments use by Chandra et al.~\cite{CHT96} in
the original CHT proof deriving $\Omega$ from any algorithm solving
consensus (to get a more detailed survey of the proof please rever
to Appendix~\ref{app:cht} or~\cite[Chapter 3]{fd-survey}).
The basic observation there is that
	a run of any algorithm using a failure detector
	induces a \emph{directed acyclic graph} (DAG).
The DAG contains a sample of failure detector values output by $\D$ in the current run
	and captures causal relations between them.
Each process $p_i$ maintains  a local copy of the DAG, denoted by $G_i$:
$p_i$ periodically queries its failure detector module, updates $G_i$
by connecting every vertex of the DAG with the vertex containing the
returned failure-detector value with an edge, and broadcasts the DAG.
An edge from vertex $[p_i,d,m]$ to vertex $[p_j,d',m']$
        is thus interpreted as ``$p_i$ queried $\D$ for the $m$th time and
        obtained value $d$ and after that $p_j$ queried $\D$ for the
        $m'$th time and obtained value $d'$''.
Whenever $p_i$ receives a DAG $G_j$ calculated earlier by $p_j$, $p_i$
merges $G_i$ with $G_j$.
As a result, DAGs maintained by the correct processes converge to the same
infinite DAG $G$.
The DAG $G_i$ is then used by $p_i$ to simulate a number of runs of
the given consensus algorithm $\A$ for all possible inputs to the processes.
All these  runs are organized in the form of a \emph{simulation
  tree} $\Upsilon_i$.
The simulation trees $\Upsilon_i$ maintained by the correct
processes converge to the same infinite simulation tree $\Upsilon$.

The outputs produced in the simulated runs of $\Upsilon_i$ are then used by $p_i$
to compute the current estimate of $\Omega$.
Every vertex $\sigma$ of $\Upsilon_i$ is assigned a valency tag based on the
decisions taken in all its \emph{extensions} (descendants of $\sigma$ in
$\Upsilon_i$): $\sigma$ is assigned a tag
$v\in\{0,1\}$ if  $\sigma$ has an extension in which some process decides
$v$.
A vertex is bivalent if it is assigned both $0$ and $1$.
It is then shown in~\cite{CHT96} that by locating the same bivalent vertex in the
limit tree $\Upsilon$, the correct process can eventually extract
the identifier of the same correct process.
(More details can be found in Appendix~\ref{app:cht} and~\cite{CHT96,fd-survey}.)

We show that this method, originally designed for
consensus, can be extended to eventual consensus (i.e., to the  weaker {\EC} abstraction). The extension is not trivial 
and requires carefully adjusting the notion of valency of a vertex in the
simulation tree.

\begin{lemma}\label{lem:NecessityOmegaEC}
In every environment $\mathcal{E}$, if a failure detector $\mathcal{D}$ implements \EC~in $\mathcal{E}$, then $\Omega$ is weaker than $\mathcal{D}$ in $\mathcal{E}$.
\end{lemma}

\begin{proof}
Let $\mathcal{A}$ be any algorithm that implements \EC~using
a failure detector $\mathcal{D}$ in an environment $\mathcal{E}$.
As in~\cite{CHT96}, every process $p_i$ maintains a failure detector sample
stored in DAG $G_i$ and periodically uses $G_i$ to simulate a set of
runs of $\A$ for all possible sequence of inputs of {\EC}.
%In each simulated run, after a process returns a response to
%$\textit{proposeEC}_{\ell}$ it, possibly after several steps,
%proceeds to execute $\textit{proposeEC}_{\ell+1}(v)$ for some $v\in\{0,1\}$.
The simulated runs are organized by $p_i$ in an ever-growing \emph{simulation
tree} $\Upsilon_i$.
A vertex of $\Upsilon_i$ is the schedule of a finite run of $\A$
``triggered'' by a path in $G_i$ in which every process starts with invoking
$\textit{proposeEC}_{1}(v)$, for some $v\in\{0,1\}$, takes steps using the failure detector
values stipulated by the path in $G_i$ and, once $\textit{proposeEC}_{\ell}(v)$
is complete, eventually invokes
$\textit{proposeEC}_{\ell+1}(v')$, for some $v'\in\{0,1\}$.
(For the record, we equip each vertex of $\Upsilon_i$ with the path in $G_i$ used to produce it.)
A vertex is connected by an edge to each one-step extension
of it.~\footnote{In~\cite{CHT96},  the simulated schedules form
  a \emph{simulation forest}, where a distinct simulation tree
  corresponds to each initial configuration encoding
  consensus inputs. Here we
  follow~\cite{JT08}: there is a single initial configuration and inputs
  are encoded in the form of input histories. As a result, we get a
  single simulation tree where branches depend on the parameters of
  $\textit{proposeEC}_{\ell}$ calls.}

Note that in every admissible infinite simulated run,
\EC-Termination, \EC-Integrity and \EC-Validity
are satisfied and that there is $k>0$
such that for all $\ell\geq k$, the invocations and responses of
$\textit{proposeEC}_{\ell}$ satisfy the \EC-Agreement.

Since processes periodically broadcast their DAGs,
the simulation tree $\Upsilon_i$ constructed locally by a correct process $p_i$
\emph{converges} to an infinite simulation tree $\Upsilon$,
in the sense that every finite subtree of
$\Upsilon$  is eventually part of $\Upsilon_i$.
The infinite simulation tree $\Upsilon$, starting
from the initial configuration of $\mathcal{A}$ and, in the limit, contains
all possible schedules that can triggered by the paths DAGs $G_i$.

Consider a vertex $\sigma$ in $\Upsilon$ identifying a unique finite
schedule of a run  of $\mathcal{A}$ using $\mathcal{D}$ in the current failure
pattern $F$. For $k>0$, we say that $\sigma$ is \emph{$k$-enabled} if $k=1$ or $\sigma$
contains a response from $\textit{proposeEC}_{k-1}$ at some process.
Now we associate each vertex $\sigma$ in $\Upsilon$ with a set of
\emph{valency tags} associated with each ``consensus instance''  $k$,
called the \emph{$k$-tag} of $\sigma$, as follows:

\begin{itemize}
\item If $\sigma$ is $k$-enabled and has a
descendant (in $\Upsilon$) in which $\textit{proposeEC}_k$ returns $x\in\{0,1\}$, then
$x$ is added to the $k$-tag of $\sigma$.

\item If $\sigma$ is $k$-enabled and has a descendant in which two
different values are returned by $\textit{proposeEC}_{k}$, then $\bot$ is added to
the $k$-tag of $\sigma$.

\end{itemize}

If $\sigma$ is not $k$-enabled, then its $k$-tag is empty.
If the $k$-tag of $\sigma$ is $\{x\}$, $x\in\{0,1\}$, we say that
$\sigma$  is \emph{$(k,x)$-valent} (\emph{$k$-univalent}).  If the
$k$-tag is $\{0,1\}$,  then we say that $\sigma$ is $k$-\emph{bivalent}.
If the $k$-tag of $\sigma$ contains $\bot$, we say that $\sigma$ is \emph{$k$-invalid}

Since $\mathcal{A}$ ensures \EC-Termination in all admissible runs extending $\sigma$,
each $k$-enabled vertex $\sigma$, the $k$-tag of $\sigma$ is non-empty.
Moreover, \EC-Termination and \EC-Validity imply that a vertex
in which no process has invoked $\textit{proposeEC}_{k}$ yet has a
descendant in which $\textit{proposeEC}_{k}$ returns $0$ and a
descendant in which $\textit{proposeEC}_{k}$ returns $1$.
Indeed, a run in which only $v$, $v\in\{0,1\}$ is proposed in instance
$k$ and every correct process takes enough steps must contain $v$ as
an output.
Thus:

\begin{enumerate}
\item[(*)] For each vertex $\sigma$, there exists $k\in\mathbb{N}$ and $\sigma'$,
a descendant of $\sigma$, such that $k$-tag of $\sigma'$ contains $\{0,1\}$.
\end{enumerate}

\begin{algorithm}[t]
\caption{Locating a bivalent vertex in $\Upsilon$.}\label{alg:bivalent}
\footnotesize
$k := 1$\\
$\sigma := $ root of $\Upsilon$\\
while true do\\
$\left\lfloor\begin{array}{l}
\text{if } \sigma \text{ is } k \text{-bivalent then break}\\
%\left\lfloor\begin{array}{l}
%\text{break}
%\end{array}\right.\\
\sigma_1 := \text{ a descendant of } \sigma \text{ in which}\\
\hspace{1cm} \text{\EC-Agreement does not hold for }\textit{proposeEC}_k\\
\sigma_2 := \text{ a descendant of } \sigma_1 \text{ in which every correct process}\\
\hspace{1cm} \text{completes } \textit{proposeEC}_k \text{ and
  receives }  \\
\hspace{1cm} \text{all messages sent to it in } \sigma\\
\text{choose } k'>k \text{ and } \sigma_3 \text{, a descendant of }
\sigma_2 \text{, such that }\\
\hspace{1cm} k' \text{-tag of } \sigma_3 \text{ contains } \{0,1\}\\
k := k'\\
\sigma := \sigma_3\\
\end{array}\right.$
\normalsize
\end{algorithm}

If the ``limit tree'' $\Upsilon$ contains a $k$-bivalent vertex, we
can apply the arguments of~\cite{CHT96} to extract $\Omega$. % (see Appendix~\ref{app:cht}).
Now we show that such a vertex exists in $\Upsilon$.  Then we
can simply let every process locate the ``first'' such vertex in its
local tree $\Upsilon_i$.
To establish an order on the vertices, we can associate
each vertex $\sigma$ of $\Upsilon$ with the value $m$
such that vertex $[p_i,d,m]$ of $G$ is used to
simulate the last step of $\sigma$ (recall that we equip
each vertex of $\Upsilon$ with the corresponding path).
Then we order vertices of $\Upsilon$ in the order consistent with the
growth of $m$. Since every vertex in $G$ has only finitely many
incoming edges, the sets of vertices having the same value of $m$ are
finite.
Thus, we can break the ties in the $m$-based order using any
deterministic procedure on these finite sets.

Eventually, by choosing the first $k$-bivalent vertex in their local
trees $\Upsilon_i$, the correct processes will eventually stabilize on the same
$k$-bivalent vertex $\tilde\sigma$ in the limit tree $\Upsilon$
and apply the CHT extraction procedure to
derive the same correct process based on $k$-tags assigned to
$\tilde\sigma$'s descendants.

It remains to show  that $\Upsilon$ indeed contains a $k$-bivalent
vertex for some $k$.  Consider the procedure described in
Algorithm~\ref{alg:bivalent} that intends to locate such a  vertex, starting with the root of the tree.

For the currently considered $k$-enabled vertex $\sigma$ that is \emph{not}
$k$-bivalent (if it is $k$-bivalent, we
are done), we use (*) to locate $\sigma_3$, a descendant of $\sigma$, such that
(1)~in $\sigma_3$, two processes return different values in $\textit{proposeEC}_k$ in
$\sigma_3$,
(2)~in $\sigma_3$, every correct process has completed $\textit{proposeEC}_k$ and
has received every message sent to it in $\sigma$, and
(3)~the $k'$-tag of $\sigma_3$ contains $\{0,1\}$.

Thus, the procedure in Algorithm~\ref{alg:bivalent}  either terminates by locating a $k$-bivalent tag and then we
are done, or it never terminates. Suppose, by contradiction, that the procedure
never terminates. Hence, we have an
infinite admissible run of $\mathcal{A}$ in which no agreement is
provided in infinitely many instances of consensus.  Indeed, in the
constructed path along the tree, every correct process appears
infinitely many times and receives every message sent to it.
This admissible run violated the \EC-Agreement property of \EC---a contradiction.

Thus, the correct processes will eventually locate the same
$k$-bivalent vertex and then, as in~\cite{CHT96},  stabilize extracting the same correct process identifier to emulate $\Omega$.
\end{proof}

%\subsection{$\Omega$ is Sufficient for \EC}
\paragraph{$\Omega$ is sufficient for \EC}
Chandra and Toueg proved that $\Omega$ is sufficient to implement the
classical version of the consensus abstraction in an environment
where a majority of processes are correct \cite{CT96}. In this
section, we extend this result to the eventual consensus abstraction for any environment.

The proposed implementation of \EC~is very simple. Each process has access to an $\Omega$ failure detector module.
Upon each invocation of the \EC~primitive, a process broadcasts the
proposed value (and the associated consensus index).
Every process stores every received value.
Each process $p_i$ periodically checks whether it has received a value for the current consensus instance from the process
that it currently believes to be the leader.  If so, $p_i$ returns this value.
The correctness of this \EC~implementation relies on the fact that,
eventually, all correct processes trust the same leader
(by the definition of  $\Omega$) and then decide (return responses) consistently on the values proposed by this process. 
% RG->PK: not sure this is needed Note that this \EC~implementation requires only one communication step per consensus instance.

\begin{lemma}\label{lem:SufficiencyOmegaEC}
In every environment $\mathcal{E}$, \EC~can be implemented using $\Omega$.
\end{lemma}
\begin{proof}
We propose such an implementation in Algorithm~\ref{algo:EC}. Then, we
prove that any admissible run $r$ of the algorithm in any environment
$\mathcal{E}$
satisfies the \EC-Termination, \EC-Integrity, \EC-Agreement, and \EC-Validity properties.

% Termination
Assume that a correct process never returns from an invocation of
$\textit{proposeEC}$ in $r$. Without loss of generality, denote by
$\ell$ the smallest integer such that a correct process $p_i$ never
returns from the invocation of $\textit{proposeEC}_\ell$. This implies
that $p_i$ always evaluates $received_i[\Omega_i,count_i]$ to
$\bot$. We know by definition of $\Omega$ that, eventually, $\Omega_i$
always returns the same correct process $p_j$. Hence, by construction
of $\ell$, $p_j$ returns from $\textit{proposeEC}_{0}$,...,
$\textit{proposeEC}_{\ell-1}$ and then sends the message
$promote(v,\ell)$ to all processes in a finite time. As $p_i$ and
$p_j$ are correct, $p_i$ receives this message and updates
$received_i[\Omega_i,count_i]$ to $v$ in a finite time. Therefore, the
algorithm satisfies the \EC-Termination property.

% Integrity
The update of the variable $count_i$ to $\ell$ for any process $p_i$
that invokes $\textit{proposeEC}_\ell$ and the assumptions on operations
$\textit{proposeEC}$ ensure us that $p_i$ executes at most once the
function \linebreak $DecideEC(\ell,received_i[\Omega_i,\ell])$. Hence, the
\EC-Integrity property is satisfied.

% Definition of k
Let $\tau_\Omega$ be the time from which the local outputs of $\Omega$
are identical and constants for all correct processes in $r$. Let $k$
be the smallest integer such that any process that invokes
$\textit{proposeEC}_k$ in $r$ invokes it after $\tau_\Omega$.

% Agreement
Let $\ell$ be an integer such that $\ell\geq k$. Assume that $p_i$ and
$p_j$ are two processes that respond to $\textit{proposeEC}_\ell$. Then,
they respectively execute the function
$DecideEC(\ell,received_i[\Omega_i,\ell])$ and
$DecideEC(\ell,received_j[\Omega_j,\ell])$. By construction of $k$, we
can deduce that $\Omega_i=\Omega_j=p_l$. That implies that $p_i$ and
$p_j$ both received a message $promote(v,\ell)$ from $p_l$. As $p_l$
sends such a message at most once, we can deduce that
$received_i[p_l,\ell]=received_j[p_l,\ell]$, that proves that ensures
the \EC-Agreement property.

% Validity
Let $\ell$ be an integer such that $\ell\geq k$. Assume that $p_i$ is
a process that respond to $\textit{proposeEC}_\ell$. The value returned
by $p_i$ was previously received from $\Omega_i$ in a message of type
$promote$. By construction of the protocol, $\Omega_i$ sends only one
message of this type and this latter contains the value proposed to
$\Omega_i$, hence, the \EC-Validity property is satisfied.

Thus, Algorithm~\ref{algo:EC} indeed implements \EC~in any
environment using $\Omega$.
\end{proof}

\begin{algorithm}
\caption{\EC~using $\Omega$: algorithm for process $p_i$}\label{algo:EC}
\small
\begin{description}\itemsep0pt
\item[Local variables:]~\\
$count_i$: integer (initially $0$) that stores the number of the last instances of consensus invoked by $p_i$\\
$received_i$: two dimensional tabular that stores a value for each pair of processes/integer (initially $\bot$)

\item[Functions:]~\\
$DecideEC(\ell,v)$ returns the value $v$ as a response to $\textit{proposeEC}_{\ell}$

\item[Messages:]~\\
$promote(v,\ell)$ with $v\in\{0,1\}$ and $\ell\in\mathbb{N}$

\item[On invocation of $\textit{proposeEC}_{\ell}(v)$]~\\
$count_i:=\ell$\\
$Send$ $promote(v,\ell)$ to all

\item[On reception of $promote(v,\ell)$ from $p_j$]~\\
$received_i[j,\ell]:=v$

\item[On local time out]~\\
If $received_i[\Omega_i,count_i]\neq\bot$ do\\
$\left\lfloor\begin{array}{l}
DecideEC(count_i,received_i[\Omega_i,count_i])
\end{array}\right.$

\end{description}
\normalsize
\end{algorithm}

%(The proof is delegated to Appendix~\ref{sec:SufficiencyOmegaEC}.)

\section{An Eventual Total Order Broadcast Algorithm}
\label{sec:ETOB}

We have shown in the previous section that $\Omega$ is the weakest
failure detector for the {\EC} abstraction (and, by Theorem~\ref{th:EquivalenceECETOB}, the
{\ETOB} abstraction) in any environment.
In this section, we describe an algorithm that directly implements \ETOB~using
$\Omega$ and which we believe is interesting in its own right. 

The algorithm has  three interesting properties.   First, it
needs only two communication steps to deliver any message when the leader does not change,  whereas algorithms implementing
classical \TOB~need at least three communication steps in this case.
Second,  the algorithm actually implements total order broadcast if
$\Omega$ outputs the same leader at all processes from the very beginning.
Third, the algorithm additionally ensures the property of  \TOB-Causal-Order,
which does not require more information about faults.

%In the algorithm, each process has access to an $\Omega$ failure
%detector module. 
The intuition behind this algorithm is as follows. Every process that intends to \ETOB-broadcast
a message sends it to all other processes.
Each process $p_i$ has access to an $\Omega$ failure
detector module and maintains a DAG that stores the set of messages
delivered so far together with their causal dependencies.
As long as $p_i$ considers itself the leader (its module of $\Omega$
outputs $p_i$), it periodically sends to all processes a sequence of
messages computed from its DAG so that the sequence respects the causal order and admits the last delivered sequence as a prefix.
%[[PK messages do not satisfy the property, TOB protocols do
%satisfies the \TOB-Causal-Order property.
%]]
A process that receives a sequence of messages delivers it only if it has 
been sent by the current leader output by $\Omega$.
The correctness of this algorithm directly follows from the properties
of $\Omega$. Indeed, once all correct processes trust the same leader,
this leader promotes its own sequence of messages, which 
ensures the \ETOB~specification. 

The pseudocode of the algorithm  is given in Algorithm
\ref{algo:ETOB}).
Below we present the proof of its correctness,  including the proof that the algorithm
additionally ensures \TOB-Causal-Order.

\begin{algorithm}[t]
\caption{$\mathcal{ETOB}$: protocol for process $p_i$}\label{algo:ETOB}
\footnotesize
\begin{description}\itemsep0pt

\item[Output variable:]~\\
$d_i$: sequence of messages $m\in M$ (initially empty) output by $p_i$

\item[Internal variables:]~\\
$promote_i$: sequence of messages $m\in M$ (initially empty) promoted by $p_i$ when $\Omega_i=p_i$\\
$CG_i$: directed graph on messages of $M$ (initially empty) that contains causality dependencies known by $p_i$

\item[Messages:]~\\
$update(CG_i)$  with $CG_i$ a directed graph on messages of $M$\\
$promote(promote_i)$ with $promote_i$ a sequence of messages $m\in M$

\item[Functions:]~\\
$UpdateCG(m,C(m))$ adds the node $m$ and the set of edges $\{(m',m)|m'\in C(m)\}$ to $CG_i$\\
$UnionCG(CG_j)$ replaces $CG_i$ by the union of $CG_i$ and $CG_j$\\
$UpdatePromote()$ replaces $promote_i$ by one of the sequences of messages $s$ such that $promote_i$ is a prefix of $s$, $s$ contains once all messages of $CG_i$, and for every edge $(m_1,m_2)$ of $CG_i$, $m_1$ appears before $m_2$ in $s$

\item[On $broadcastETOB(m,C(m))$ from the application]~\\
$UpdateCG(m,C(m))$\\
$Send$ $update(CG_i)$ to all

\item[On reception of $update(CG_j)$ from $p_j$]~\\
$UnionCG(CG_j)$\\
$UpdatePromote()$

\item[On reception of $promote(promote_j)$ from $p_j$]~\\
If $\Omega_i=p_j$ then\\
$\left\lfloor\begin{array}{l}
d_i:=promote_j
\end{array}\right.$

\item[On local time out]~\\
If $\Omega_i=p_i$ then\\
$\left\lfloor\begin{array}{l}
Send \text{ } promote(promote_i) \text{ to all}
\end{array}\right.$

\end{description}
\normalsize
\end{algorithm}

\begin{lemma}\label{lem:etob}
In every environment $\E$, Algorithm $\mathcal{ETOB}$ implements \ETOB~using $\Omega$. 
\end{lemma}
\begin{proof}
First, we prove that any run $r$ of $\mathcal{ETOB}$ in any environment $\mathcal{E}$ satisfies the \TOB-Validity, \TOB-No-creation, \TOB-No-duplication, and \TOB-Agreement properties.

%Validity
Assume that a correct process $p_i$ broadcasts a message $m$ at time $t$ for a given $t\in\mathbb{N}$. We know that $\Omega$ outputs the same correct process $p_j$ to all correct processes in a finite time. As $p_j$ is correct, it receives the message $update(CG_i)$ from $p_i$ (that contains $m$) in a finite time. Then, $p_j$ includes $m$ in its causality graph (by a call to $UnionCG$) and in its promotion sequence (by a call to $UpdatePromote$). As $p_j$ never removes a message from its promotion sequence and is outputted by $\Omega$, $p_i$ adopts the promotion sequence of $p_j$ in a finite time and this sequence contains $m$, that proves that $\mathcal{ETOB}$ satisfies the \TOB-Validity property.

%No-creation
Any sequence outputted by any process is built by a call to $UpdatePromote$ by a process $p_i$.  This function ensures that any message appearing in the computed sequence appears in the graph $CG_p$. This graph is built by successive calls to $UnionCG$ that ensure that the graph contains only messages received in a message of type $update$. The construction of the protocol ensures us that such messages have been broadcast by a process. Then, we can deduce that $\mathcal{ETOB}$ satisfies the \TOB-No-creation property.

%No duplication
Any sequence outputted by any process is built by a call to $UpdatePromote$ that ensures that any message appears only once. Then, we can deduce that $\mathcal{ETOB}$ satisfies the \TOB-No-duplication property.

%Agreement
Assume that a correct process $p_i$ stably delivers a message $m$ at time $t$
for a given $t\in\mathbb{N}$. We know that $\Omega$ outputs the same
correct process $p_j$ to all correct processes after some finite
time. Since $m$ appears in every $d_i(t')$ such that $t'\geq t$, we
derive that $m$ appears infinitely in $promote_j$ from a given point
of the run. Hence, the construction of the protocol and the correctness
of $p_j$ implies that any correct process eventually stably delivers
$m$, and  $\mathcal{ETOB}$ satisfies the \TOB-Agreement property.

We now prove that, for any environment $\mathcal{E}$, for any run $r$ of $\mathcal{ETOB}$ in $\mathcal{E}$, there exists a $\tau\in\mathbb{N}$ satisfying \ETOB-Stability, \ETOB-Total-order, and \TOB-Causal-Order properties in $r$. Hence, let $r$ be a run of $\mathcal{ETOB}$ in an environment $\mathcal{E}$. Let us define:
\begin{itemize}
\item $\tau_\Omega$ the time from which the local outputs of $\Omega$ are identical and constant for all correct processes in $r$;
\item $\Delta_c$ the longest communication delay between two correct processes in $r$;
\item $\Delta_t$ the longest local timeout for correct processes in $r$;
\item $\tau=\tau_\Omega+\Delta_t+\Delta_c$
\end{itemize}

%Stability
Let $p_i$ be a correct process and $p_j$ be the correct elected by $\Omega$ after $\tau_\Omega$. Let $t_1$ and $t_2$ be two integers such that $\tau\leq t_1\leq t_2$. As the output of $\Omega$ is stable after $\tau_\Omega$ and the choice of $\tau$ ensures us that $p_i$ receives at least one message of type $promote$ from $p_j$, we can deduce from the construction of the protocol that there exists $t_3\leq t_1$ and $t_4\leq t_2$ such that $d_i(t_1)=promote_j(t_3)$ and $d_i(t_2)=promote_j(t_4)$. But the function $UpdatePromote$ used to build $promote_j$ ensures that $promote_j(t_3)$ is a prefix of $promote_j(t_4)$. Then, $\mathcal{ETOB}$ satisfies the \ETOB-Stability property after time $\tau$.

%Total order
Let $p_i$ and $p_j$ be two correct processes such that two messages $m_1$ and $m_2$ appear in $d_i(t)$ and $d_j(t)$ at time $t\geq \tau$. Assume that $m_1$ appears before $m_2$ in $d_i(t)$. Let $p_k$ be the correct elected by $\Omega$ after $\tau_\Omega$. As the output of $\Omega$ is stable after $\tau_\Omega$ and the choice of $\tau$ ensures us that $p_i$ and $p_j$ receive at least one message of type $promote$ from $p_j$, the construction of the protocol ensures us that we can consider $t_1$ and $t_2$ such that $d_i(t)=promote_k(t_1)$ and $d_j(t)=promote_k(t_2)$. The definition of the function $UpdatePromote$ executed by $p_k$ allows us to deduce that either $d_i(t)$ is a prefix of $d_j(t)$ or $d_j(t)$ is a prefix of $d_i(t)$. In both cases, we obtain that $m_1$ appears before $m_2$ in $d_j(t)$, that proves that $\mathcal{ETOB}$ satisfies the \ETOB-Total-order property after time $\tau$.

%Causal order
Let $p_i$ be a correct process such that two messages $m_1$ and $m_2$ appear in $d_i(t)$ at time $t\geq 0$. Assume that $m_1\in C(m_2)$ when $m_2$ is broadcast. Let $p_j$ be the process trusted by $\Omega_i$ at the time $p_i$ adopts the sequence $d_i(t)$. If $m_2$ appears in $d_i(t)$, that implies that the edge $(m_1,m_2)$ appears in $CG_j$ at the time $p_j$ executes $UpdatePromote$ (since $p_j$ previously executed $UnionCG$ that includes at least $m$ and the set of edges $\{(m',m)|m'\in C(m)\}$ in $CG_j$). The construction of $UpdatePromote$ ensures us that $m_1$ appears before $m_2$ in $d_i(t)$, that proves that $\mathcal{ETOB}$ satisfies the \TOB-Causal-Order property.

In conclusion, $\mathcal{ETOB}$ is an implementation of \ETOB~assuming that processes have access to the $\Omega$ failure detector in any environment.
\end{proof}

%==========================================
\section{Related Work}
\label{sec:related}
%==========================================

Modern data service providers such as Amazon's Dynamo
\cite{DeCandia2007},  Yahoo's PNUTs~\cite{Cooper2008} or Google
Bigtable distributed storage~\cite{Chang2008} are intended to offer
highly available services. They consequently replicate those services
over several server processes. In order to  tolerate process failures as well
as partitions, they consider eventual consistency~\cite{SS05,
  Vogels2009,Singh2009}. 
  
The term {\em eventual} consensus was introduced in~\cite{KMO11}. It
refers to one instance of consensus which stabilizes at the end; not multiple instances as we consider in this paper.
% It refers to the classical definition of consensus in the precise sense that the processes decide
% asynchronously by opposition to {\em simultaneous} consensus in which the processes are
%  required to output their decision value simultaneously.
In~\cite{Dolev2010}, a self-stabilizing form of consensus was proposed: 
assuming a self-stabilizing implementation of $\diamond S$ (also described in the paper) and executing a sequence of consensus instances,   validity and  agreement are eventually ensured. Their consensus abstraction is close to ours but the authors focused on the shared-memory model and did not address the question of the weakest failure detector.

In~\cite{Fekete1996}, the intuition behind eventual consistency was
captured through the concept of eventual serializability. Two kinds of operations
were defined: (1) a
``stable'' operation of which response needs to be totally ordered
after  all operations preceding it and (2) ``weak'' operations of
which responses might not reflect all their preceding operations.
Our {\ETOB} abstraction captures consistency with respect to the
``weak'' operations. (Our lower bound on the necessity of 
$\Omega$ naturally extends to the stronger definitions.)

Our perspective on eventual consistency is closely related to the
notion of \emph{eventual linearizability} discussed recently
in~\cite{Serafini2010} and~\cite{GR14}.
%[[PK I think it's just their mistake in the formulation
%The definition given in~\cite{Serafini2010}, however, assumes a global
%stabilization time to hold in \emph{all} runs of an eventually
%linearizable implementation, which essentially turns their version of eventual
%linearizability into a safety property, and not a liveness property as
%one would expect given the name ``eventual''.
%]]
It is shown in~\cite{Serafini2010} that the weakest failure detector
to boost eventually linearizable objects to linearizable ones is
$\Diamond P$.  We are focusing primarily on the weakest failure
detector to \emph{implement} eventual
%[[PK we do not call it linearizability
consistency,
%linearizability,
so their result is orthogonal to ours.
% (besides the discrepancy in the definition).
%]]

%An alternative definition of eventual linearizability was also given in~\cite{GR14}.
%The stabilization time is defined per run and, naturally, no
%bound on it might exist in an eventually linearizable implementation.
%Our {\ETOB} specification is close to this definition.
%]]

%A similar distinction was considered in \cite{Serafini2010} where
%``weak'' operations are eventually linearizable. The definition of
%eventual linearizability is based
%on a prefix consistency property of operations, \emph{i.e.}, at any
%time there
%exists a common prefix gathering the events associated to all correct
%processes. An alternative definition of eventual linearizability was also given in \cite{GR14}.
%
In~\cite{GR14}, eventual linearizability is
compared against linearizability in the context of implementing
specific objects in a shared-memory context.
It turns out that an eventually linearizable implementation of a \textit{fetch-and-increment}
object is as hard to achieve as a linearizable one.
Our {\ETOB} construction can be seen as an \emph{eventually
linearizable universal construction}: given any sequential
object type, {\ETOB} provides an eventually linearizable
concurrent implementation of it.
Brought to the message-passing environment with a correct majority,
our results complement~\cite{GR14}:
we show that in this setting, an eventually consistent replicated
service  (eventually linearizable object with a sequential
specification) requires exactly the same information about failures
as a consistent (linearizable) one.

%==========================================

\section{Concluding Remarks}

\label{sec:conclu}

%==========================================

This paper defined the abstraction of eventual total order broadcast and proved its
equivalence to eventual consensus: two fundamental building
blocks to implement a general replicated state machine that ensures eventual consistency.
We proved that the weakest failure detector to implement
these abstractions is $\Omega$, in any message-passing environment.
We could hence determine the gap between
building a general replicated state machine that ensures consistency in a message-passing system and one
that ensures only eventual consistency.
In terms of information about failures, this gap is precisely captured
by failure detector $\Sigma$~\cite{DFG10}. In terms of time
complexity, the gap is exactly one message delay: an operation on the strongly
consistent replicated must, in the worst case, incur three
communication steps~\cite{Lam06}, while one build using our eventually total order 
broadcast protocol  completes an operation in the optimal number of two communication steps.
%It would
%be interesting to complement our result by expressing the gap in terms of a complexity result (e.g., minimal
%number of messages or communication steps needed to deliver a message in ``good'' runs).

Our {\ETOB} abstraction captures a form of eventual consistency
implemented in multiple replicated
services~\cite{DeCandia2007,Cooper2008,Chang2008}. In addition  to
eventual consistency guarantees, such systems sometimes produce
indications when a prefix of operations on the replicated service is
\emph{committed}, i.e., is not subject to further changes.
A prefix of operations can be committed, \emph{e.g.}, in sufficiently
long periods of synchrony, when a majority of correct processes elect
the same leader and  all incoming and outgoing messages of
the leader to the correct majority are delivered within some fixed bound.
We believe that such indications could easily be implemented, during
the stable periods, on top of {\ETOB}. 
Naturally, our results imply that $\Omega$ is necessary for such systems too. 

Our {\EC} abstraction assumes eventual agreement, but requires
integrity and validity to be always ensured.
Other definitions of eventual consensus could be considered. In
particular, we have studied an eventual
consensus abstraction assuming, instead of eventual agrement, \emph{eventual integrity}, \emph{i.e.}, a bounded number of decisions
in a given consensus instance could be revoked a finite
number of times. In Appendix~\ref{app:eic}, 
we define this abstraction 
of eventual \emph{irrevocable} consensus  (\EIC) more precisely and show 
that it is equivalent to our  {\EC} abstraction.

%==========================================
%\section{Conclusion}
%\label{sec:conclu}
%==========================================

%In this paper, we presented a formalization of eventual consistency and explored the theoretical
%  assumptions needed to implement this. We gave precise definitions of eventual consensus and eventual
%  total order broadcast and related them to the notion of eventual consistency. We showed that these two
%  abstractions are equivalent. We established then that $\Omega$ failure detector is necessary and sufficient
%  for eventual consensus in any environment. An interesting and surprisingly aspect is that we do not
%  assume a majority of correct processes while classical consistency and consensus require the quorum
% failure detector $\Sigma$. This theoretical result reinforces the interest of eventual consistency
% for partition-tolerance.

% We considered eventual agreement and eventual validity in the consensus abstraction. It is notable that
%  other definitions of eventual consensus could be considered. We especially have studied an eventual
%  consensus where the integrity property eventually holds, \emph{i.e.}, a decision could be revoked a finite
%  number of times. In Appendix~\ref{app:eic}, we define the eventual irrevocable consensus that this
% abstraction is equivalent to the eventual consensus abstraction considered in this paper.

%==========================================
%\newpage

%\bibliographystyle{abbrv}
%{\small
%\bibliography{references}
%}

\def\noopsort#1{} \def\No{\kern-.25em\lower.2ex\hbox{\char'27}}
  \def\no#1{\relax} \def\http#1{{\\{\small\tt
  http://www-litp.ibp.fr:80/{$\sim$}#1}}}

\newpage

\appendix
\section{Discussion on Eventual Consensus}
\label{app:eic}
%==========================================

Our definition of Eventual Consensus \EC~relaxes the Agreement property which holds after a finite number of operations. We could instead relax the Integrity property where processes can change their decisions a finite number of times. We discuss here the resulting abstraction.

\subsection{Eventual Irrevocable Consensus (\EIC)}

The \emph{eventual irrevocable consensus} (\EIC) abstraction exports, to every process $p_i$, operations $\textit{proposeEIC}_{0}$, $\textit{proposeEIC}_{1}$, $\ldots$ that take binary arguments and return binary responses. If a process $p_i$ responds more than once to $\textit{proposeEIC}_{\ell}$ for some $\ell\in\mathbb{N}$, we consider that the response of $p_i$ to $\textit{proposeEIC}_{\ell}$ at time $t\in\mathbb{N}$ is its last response to $\textit{proposeEIC}_{\ell}$ before $t$.

Assuming that every process receives $\textit{proposeEIC}_{\ell}$ as soon as it returns a (first) response to $\textit{proposeEIC}_{\ell-1}$ for all $\ell \in\mathbb{N}$, the abstraction guarantees, for every run, there exists $k\in\mathbb{N}$ such that the following properties are satisfied:

\begin{description}\itemsep0pt
\item[\EIC-Termination] Every correct process eventually returns a response to $\textit{proposeEIC}_{\ell}$ for all $\ell\in\mathbb{N}$.
\item[\EIC-Integrity] No process responds twice to $\textit{proposeEIC}_{\ell}$ for all $\ell\geq k$.
\item[\EIC-Agreement] No two processes return infinitely different values to $\textit{proposeEIC}_{\ell}$ for any $\ell\in\mathbb{N}$.
\item[\EIC-Validity] Every value returned to $\textit{proposeEIC}_{j}$ was previously proposed to $\textit{proposeEIC}_{j}$ for all $j\in\mathbb{N}$.
\end{description}

\begin{theorem}\label{th:EquivalenceECEIC}
In every environment $\mathcal{E}$, \EC~and \EIC~are equivalent.
\end{theorem}

\subsection{Transformation from \EC~to \EIC}

\begin{lemma}\label{lem:ECtoEIC}
In every environment $\mathcal{E}$, there exists a transformation from \EC~to \EIC.
\end{lemma}

\begin{proof}
To prove this result, it is sufficient to provide a protocol that implements \EIC~in an environment $\mathcal{E}$ knowing that there exists a protocol that implements \EC~in this environment. This transformation protocol $\mathcal{T}_{\EC\rightarrow\EIC}$ is stated in Algorithm \ref{algo:ECtoEIC}. Now, we are going to prove that $\mathcal{T}_{\EC\rightarrow\EIC}$ implements \EIC.

%Termination
As any invocation of $\textit{proposeEIC}_\ell$ by a correct process $p_i$ leads to an invocation of $\textit{ProposeEC}_{\ell}$ by the same process, the \EC-Termination property ensures us that $p_i$ receives eventually a response (a sequence $decison$) from the \EC~primitive. Before this response, we have $decision_i[\ell]=\bot$. By the \EC-Validity property, we know that $decision[\ell]$ is a value proposed by one process (hence not equal to $\bot$). Then, the construction of the protocol ensures us that $DecideEIC(\ell,decision[\ell])$ is executed in a finite time, that proves that $\mathcal{T}_{\EC\rightarrow\EIC}$ satisfies the \EIC-Termination property.

Let $k$ be the index after which the \EC~primitive satisfies \EC-Agreement property. Let $\tau$ be the smallest time where all correct processes receive the response of $\textit{proposeEC}_k$.

%Integrity
After time $\tau$, we know that the sequences $decision$ returned to all process are identical. Then, the construction of the protocol ensures us that every sequence submitted to the \EC~primitive is prefixed by the last sequence returned by this primitive. Hence, the \EC-Agreement property ensures us that, after time $\tau$, $DecideEIC$ is executed only for the last value of the decision sequence and only when this sequence grows, that proves that $\mathcal{T}_{\EC\rightarrow\EIC}$ satisfies the \EIC-Integrity property.

%Agreement
Assume that two processes $p_i$ and $p_j$ return forever two different values for $\textit{proposeEIC}_\ell$ for some $\ell$. By the \EIC-Integrity property proved above, we know that $p_i$ and $p_j$ take at most one decision for $\textit{proposeEIC}_\ell$ after time $\tau$. That implies that $p_i$ and $p_j$ return different values at their last decision. Then, we can deduce that $decision_i[\ell]\neq decision_j[\ell]$ forever, that is contradictory with the definition of $\tau$ or with the \EC-Agreement property. This contradiction shows us that $\mathcal{T}_{\EC\rightarrow\EIC}$ satisfies the \EIC-Agreement property.

%Validity
The fact that $\mathcal{T}_{\EC\rightarrow\EIC}$ satisfies the \EIC-Validity directly follows from the \EC-Validity.

In conclusion, $\mathcal{T}_{\EC\rightarrow\EIC}$ satisfies the \EIC~specification in an environment $\mathcal{E}$ provided that there exists a protocol that implements \EC~in this environment.
\end{proof}

\begin{algorithm}
\caption{$\mathcal{T}_{\EC\rightarrow\EIC}$: transformation from \EC~to \EIC~for process $p_i$}\label{algo:ECtoEIC}
\small
\begin{description}\itemsep0pt
\item[Internal variables:]~\\
$decision_i$: sequence of values decided by $p_i$ (initially $\epsilon$)

\item[Functions:]~\\
$DecideEIC(\ell,v)$ returns the value $v$ as a response to $\textit{proposeEIC}_{\ell}$

\item[On invocation of $\textit{proposeEIC}_{\ell}(v)$]~\\
$\textit{proposeEC}_{\ell}(decision_i.v)$

\item[On reception of $decision$ as response of $\textit{proposeEC}_{\ell}$]~\\
For $k$ from $0$ to $\ell$ do\\
$\left\lfloor\begin{array}{l}
\text{If } decision[k]\neq decision_i[k] \text{ then}\\
\left\lfloor\begin{array}{l}
DecideEIC(k,decision[k])
\end{array}\right.

\end{array}\right.$\\
$decision_i:=decision$

\end{description}
\normalsize
\end{algorithm}

\subsection{Transformation from \EIC~to \EC}

\begin{lemma}\label{lem:EICtoEC}
In every environment $\mathcal{E}$, there exists a transformation from \EIC~to \EC.
\end{lemma}

\begin{proof}
To prove this result, it is sufficient to provide a protocol that implements \EC~in an environment $\mathcal{E}$ knowing that there exists a protocol that implements \EIC~in this environment. This transformation protocol $\mathcal{T}_{\EIC\rightarrow\EC}$ is stated in Algorithm \ref{algo:EICtoEC}. Now, we are going to prove that $\mathcal{T}_{\EIC\rightarrow\EC}$ implements \EC.

%Termination and Integrity
As any invocation of $\textit{proposeEC}_\ell$ by a correct process $p_i$ leads to an invocation of $\textit{proposeEIC}_{\ell}$ by the same process, the \EIC-Termination property ensures us that $p_i$ receives eventually at least one response from the \EIC~primitive. The use of the counter $count_i$ allows us to deduce that only the first response from the \EIC~primitive leads to a decision for $\textit{proposeEC}_\ell$ by $p_i$, that proves that $\mathcal{T}_{\EIC\rightarrow\EC}$ satisfies the \EC-Termination and the \EC-Integrity properties.

%Agreement and Validity
The construction of the protocol and the \EIC-Agreement and the \EIC-Validity properties trivially imply that $\mathcal{T}_{\EIC\rightarrow\EC}$ satisfies the \EC-Agreement and the \EC-Validity properties.

In conclusion, $\mathcal{T}_{\EIC\rightarrow\EC}$ satisfies the \EC~specification in an environment $\mathcal{E}$ provided that there exists a protocol that implements \EIC~in this environment.
\end{proof}

\begin{algorithm}
\caption{$\mathcal{T}_{\EIC\rightarrow\EC}$: transformation from \EIC~to \EC~for process $p_i$}\label{algo:EICtoEC}
\small
\begin{description}\itemsep0pt
\item[Internal variables:]~\\
$count_i$: integer (initially $0$) that stores the number of the last instance of consensus invoked by $p_i$

\item[Functions:]~\\
$DecideEC(\ell,v)$ returns the value $v$ as a response to $\textit{proposeEC}_{\ell}$

\item[On invocation of $\textit{proposeEC}_{\ell}(v)$]~\\
$count_i :=\ell$\\
$\textit{proposeEIC}_{\ell}(v)$

\item[On reception of $v$ as response of $\textit{proposeEIC}_{\ell}$]~\\
If $count_i=\ell$ then\\
$\left\lfloor\begin{array}{l}
DecideEC(\ell,v)
\end{array}\right.$

\end{description}
\normalsize
\end{algorithm}

\section{Background on the CHT proof}
\label{app:cht}

Let $\E$ be any environment, $\D$ be any {\fd} that can be used
	to solve consensus in $\E$, and $\A$ be any algorithm that solves consensus
	in $\E$ using $\D$.
We  determine a reduction algorithm $T_{\D\to\Omega}$ that, using failure
	detector $\D$ and algorithm $\A$,
	implements $\Omega$ in $\E$.
Recall that implementing $\Omega$ means outputting, at every process,
	the id of a process so that eventually,
	the id of the same correct process is output permanently at all correct
	processes.

%---------------------------------------------------------------
\subsection{Overview of the reduction algorithm}
\label{subsec:cht:overview}
%---------------------------------------------------------------

The basic idea underlying $T_{\D\to\Omega}$ is to have each process
	locally {\em simulate} the overall distributed system
	in which the processes execute several runs of $\A$
	that \emph{could have happened} in the current failure
	pattern and {\fd} history.
Every process then uses these runs to extract $\Omega$.

%Every process simulates, locally, runs of
%	algorithm $\A$ by launching threads that mimic the behavior of every
%	other process in the system running algorithm $\A$  %(Fig.~\ref{fig:simulation}).
%But how can one process $p_i$	simulate the overall system executing
%	several runs of $\A$?
In the local simulations, every process $p$ feeds algorithm $\A$
	with a set of proposed values, one for each process of the
	system. %, i.e., $p$ pretends to be every other
%	process and proposes values to $\A$'s runs.
%In fact, process $p$ proposes all possible combinations as we discuss
%	in Act 2.
Then all automata composing $\A$ are triggered
	locally by $p$ which emulates,
	for every simulated run of $\A$, the states of all processes as well as
	the emulated buffer of exchanged messages.

Crucial elements that are needed for the simulation
	are (1) the values from {\fd}s
	that would be output by $\D$ as well as
	(2) the order according to which the processes
	are taking steps.
For these elements, which we call the stimuli of
	algorithm $\A$, process $p$ periodically queries its {\fd} module
	and exchanges the {\fd} information with the other processes.
%It is important to notice that the output
%	of $\D$ is not restricted in any way.
%In fact, this output can be \emph{any value} that encodes some information
%	about failures.

The reduction algorithm $T_{\D\to\Omega}$ consists of two tasks that are run
	in parallel at every process:
	the \emph{commmuncation task} and the \emph{computation task}.
In the communication task, every process maintains ever-growing stimuli
	of algorithm $\A$ by periodically querying its {\fd} module
	and sending the output to all other processes.
In the computation task, every process periodically feeds the stimuli
	to algorithm $\A$, simulates several runs of $\A$,
	and computes the current emulated output of $\Omega$.

%---------------------------------------------------------------
\subsection{Building a DAG}
\label{subsec:cht:dag}
%---------------------------------------------------------------

The communication task of algorithm $T_{\D\to\Omega}$ is presented in Figure~\ref{fig:CHT-comm}.
%In the task, every process $p$ periodically queries its {\fd} module
%	of $\D$ and sends the output to all other processes.
Executing this task, $p$ knows more and more of the processes'
	{\fd} outputs and temporal relations between
	them.
All this information is pieced together in a single data structure,
	a directed acyclic graph (DAG) $G_p$.
Informally, every vertex $[q,d,k]$ of $G_p$ is a {\fd} value
``seen'' by $q$ in its $k$-th query of its {\fd} module.
An edge  $([q,d,k],[q',d',k'])$ can be  interpreted as
	``$q$ saw {\fd} value $d$ (in its $k$-th query)
\emph{before} $q'$ saw {\fd} value $d'$ (in its $k'$-th query)''.

\begin{figure}[tbp]
\hrule \vspace{2mm}
{\small
%\fbox{
%\begin{minipage}[t]{6in}
%\setcounter{linenumber}{0}
\begin{tabbing}
 bbb\=bb\=bb\=bb\=bb\=bb\=bb\=bb \=  \kill
\> $G_p \get $ empty graph\\
\> $k_p \get 0$\\
\> \textbf{while} true \textbf{do}\\
%\>\>    \{ receive phase \}\\
\>\>	receive message $m$\\
\smallskip
%\>\>    \{ query phase \}\\
\>\> 	$d_p \get $ query {\fd} $\D$\\
\smallskip
%\>\>    \{ send phase \}\\
\>\>	$k_p\get k_p+1$\\
\>\> \textbf{if} $m$ is of the form $(q,G_q,p)$ \textbf{then} $G_p\get G_p\cup G_q$\\
\>\>	add $[p,d_p,k_p]$ and edges from all vertices of $G_p$ to $[p,d_p,k_p]$ to $G_p$\\
\>\>	send $(p,G_p,q)$ to all $q\in\Pi$
\end{tabbing}
\hrule
%\end{minipage}
%}
}
\caption{Building a DAG: process $p$}
\label{fig:CHT-comm}
\end{figure}

DAG $G_p$ has some special properties which follow from its construction.
%Consider any run of the communication component in $F$
%	with $\D$ outputting $H$.	
Let $F$ be the current failure pattern in $\E$ and $H$
	be the current {\fd} history in $\D(F)$.
Then: %for any correct process $p$ and any time $t$
	%($x(t)$ denotes the value of variable $x$ at time $t$):
\begin{enumerate}
\item[(1)]  The vertices of $G_p$ are of the form $[q,d,k]$ where $q\in\Pi$,
$d\in \R_{\D}$ and $k\in \Nat$.
There is a mapping $\tau: \mbox{ vertices of }G_p \mapsto \Time$,
associating a time with every vertex of $G_p$,
such that:
\begin{enumerate}
\item[(a)] For any vertex $v=[q,d,k]$, $q \notin F(\tau(v))$ and $d=H(q,\tau(v))$.
That is, $d$ is the value output by $q$'s {\fd} module at time $\tau(v)$.

\item[(b)] For any edge $(v,v')$ in $G_p$, $\tau(v)<\tau(v')$.
That is, any edge in $G_p$ reflects the temporal order
	in which the {\fd} values are output.
\end{enumerate}
\item[(2)] If $v'=[q,d,k]$ and $v''=[q,d',k']$ are vertices of $G_p$, and $k<k'$, then
$(v,v')$ is an edge of $G_p$.
\item[(3)] $G_p$ is transitively closed: if $(v,v')$ and $(v',v'')$ are edges of $G_p$,
	then $(v,v'')$ is also an edge of $G_p$.
\item[(4)] For all correct processes $p$ and $q$ and all times $t$,
	there is a time $t'\geq t$, a $d\in\R_{\D}$ and a
	$k\in\Nat$ such that for every  vertex $v$ of $G_p(t)$,
	$(v,[q,d,k])$ is an edge of $G_p(t')$.\footnote{
	For any variable $x$ and time $t$, $x(t)$ denotes
	the value of $x$ at time $t$.}
\end{enumerate}

Note that properties (1)--(4) imply that, for every correct process $p$, $t\in\Time$ and $k\in\Nat$,
	there is a time $t'$ such that $G_p(t')$ contains a path
	$g=[q_1,d_1,k_1]\to[q_2,d_2,k_2]\to\ldots$, such that
	(a) every correct process appears at least $k$ times in $g$,
	and (b) for any path $g'$ in $G_p(t)$, $g'\cdot g$ is also a path in $G_p(t')$.
%Furthermore, for any process $p$ and time $t$,
%	every prefix of any path in $G_p(t)$ is also a path in $G_p(t)$.

%---------------------------------------------------------------
\subsection{Simulation trees}
\label{subsec:cht:simulation}
%---------------------------------------------------------------

%The computation task of algorithm $T_{\D\to\Omega}$ is presented in Figure~\ref{fig:comp}.
Now DAG $G_p$ can be used to simulate runs of $\A$.
Any path $g=[q_{1},d_1,k_1],$ $[q_{2},d_2,k_2],$ $\ldots,[q_{s},d_s,k_s]$
	through $G_p$ gives the order in which
	processes $q_{1}$, $q_{2},\ldots,$ $q_{s}$ ``see'', respectively,
	{\fd} values $d_1$, $d_1,d_2,\ldots,$ $d_s$.
That is, $g$ contains an activation schedule and {\fd}
	outputs for the processes to execute steps of
	$\A$'s instances.
Let $I$ be any initial configuration of $\A$.
Consider a schedule $S$ that is applicable to $I$
	and \emph{compatible with $g$},
	i.e., $|S|=s$ and $\forall k\in\{1,2,\ldots,s\}$, $S[k]=(q_k,m_k,d_k)$, where
	$m_k$ is a message addressed to $q_k$ (or the null message $\lambda$).

All schedules that are applicable to $I$ and compatible with paths in $G_p$
	can be represented as a tree $\Upsilon_p^I$, called the
	\emph{simulation tree induced by $G_p$ and $I$}.
The set of vertices of $\Upsilon_p^I$ is the set of all schedules $S$ that are applicable to $I$
	and compatible with paths in $G_p$.
The root of $\Upsilon_p^I$ is the empty schedule $S_{\bot}$.
There is an edge from $S$ to $S'$ if and only if $S'=S\cdot e$ for a step $e$;
	the edge is labeled $e$.
Thus, every vertex $S$ of $\Upsilon_p^I$ is associated with a sequence
	of steps $e_1\,e_2\,\ldots\,e_s$ consisting of labels of the edges
	on the path from $S_{\bot}$ to $S$.
In addition, every descendant of $S$ in $\Upsilon_p^I$ corresponds to 	
	an extension of $e_1\,e_2\,\ldots\,e_s$.

The construction of $\Upsilon_p^I$ implies that, for any vertex $S$ of $\Upsilon_p^I$,
	there exists a partial run
	$\tpl{F,H,I,S,T}$ of $\A$ where $F$ is the current failure pattern and
	$H\in\D(F)$ is the current {\fd} history.
%$\forall k\in\{1,\ldots,s\}$,
%	$S[k]=(q_{k},m_k,d_{k})$.
Thus, if in $S$, correct processes appear sufficiently often
	and receive sufficiently many messages sent to them,
	then every correct (in $F$) process decides in $S(I)$.

\begin{figure}[htbp]
  \centering
  \includegraphics[scale=.9]{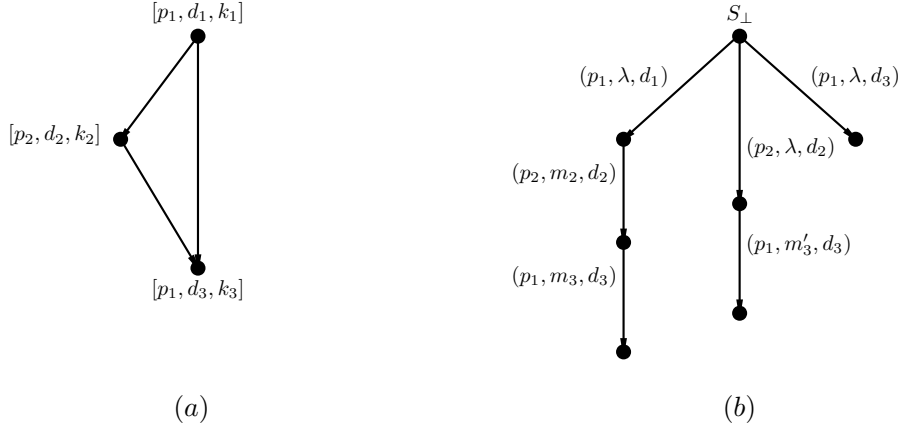}
  \caption{A DAG and a tree}
  \label{fig:dag-tree}
\end{figure}

In the example depicted in Figure~\ref{fig:dag-tree},
	a DAG (a) induces a simulation tree a portion of which is shown in (b).
There are three non-trivial paths in the DAG: $[p_1,d_1,k_1]\to[p_2,d_2,k_2]\to[p_1,d_3,k_3]$,
	$[p_2,d_2,k_2]\to[p_1,d_3,k_3]$, $[p_2,d_2,k_2]\to[p_1,d_3,k_3]$
	and $[p_1,d_1,k_1]\to[p_1,d_3,k_3]$.
Every path through the DAG and an initial configuration $I$
	induce at least one schedule in the simulation tree.
Hence, the simulation tree has at least three leaves:
	$(p_1,\lambda,d_1)$ $(p_2,m_2,d_2)$ $(p_1,m_3,d_3)$, $(p_2,\lambda,d_2)$ $(p_1,m_3',d_3)$, and
	$(p_1,\lambda,d_3)$.
Recall that $\lambda$ is the empty message: since the message buffer is empty in $I$,
	no non-empty message can be received in the first step of any schedule.

%---------------------------------------------------------------
\subsection{Tags and valences}
\label{subsec:cht:valence}
%---------------------------------------------------------------

Let $I^i$, $i\in\{0,1,\ldots,n\}$ denote the initial configuration of $\A$
	in which processes $p_1,\ldots,p_{i}$ propose $1$ and the rest (processes
	$p_{i+1},\ldots,p_n$) propose $0$.
In the computation task of the reduction algorithm, every process
	$p$ maintains an ever-growing \emph{simulation forest}
	$\Upsilon_p= \{\Upsilon_p^{0},\Upsilon_p^{1}, \ldots, \Upsilon_p^{n}\}$
	where $\Upsilon_p^{i}$ ($0\leq i\leq n$) denotes
	the simulation trees induced by $G_p$ and initial configurations
	$I^i$.

For every vertex of the simulation forest, $p$
	assigns a set of \emph{tags}.
Vertex $S$ of tree $\Upsilon_p^{i}$ is assigned a tag $v$ if and only if
	$S$ has a descendant $S'$ in $\Upsilon_p^i$ such that
	$p$ decides $v$ in $S'(I^i)$.
We call the set tags the \emph{valence} of the vertex.
By definition, if $S$ has a descendant with a tag $v$, then $S$ has tag $v$.
Validity of consensus ensures that the set of tags is a subset of $\{0,1\}$.
%Agreement of consensus ensures that no correct process can decide
%	in $S(I^i)$ for a bivalent vertex $S$.

Of course, at a given time, some vertices of the simulation forest $\Upsilon_p$ might not have
	any tags because the simulation stimuli are not sufficiently long yet.
But this is just a matter of time: if $p$ is correct, then
	every vertex of $p$'s simulation forest will eventually have an extension	
	in which correct processes appear sufficiently often for $p$
	to take a decision.
	
A vertex $S$ of $\Upsilon_p^i$ is \emph{$0$-valent} if it has exactly one tag $\{0\}$
	(only $0$ can be decided in $S$'s extensions in $\Upsilon_p^i$).
A $1$-valent vertex is analogously defined.
If a vertex $S$ has both tags $0$ and $1$ (both $0$ and $1$ can be decided
	in $S$'s extensions), then we say that $S$ is \emph{bivalent}.\footnote{
		The notion of valence was first defined in~\cite{FLP85} as
		the set of values than are decided in
		\emph{all} extensions of a given execution.
		Here we define the valence as only a subset of these values,
		defined by the simulation tree.}

It immediately follows from Validity of consensus that
	the root of $\Upsilon_p^{0}$ can at most be $0$-valent,
	and the root of $\Upsilon_p^{n}$ can at most be $1$-valent
	(the roots of $\Upsilon_p^{0}$ and $\Upsilon_p^{n}$ cannot
	be bivalent).
%Further, agreement of consensus ensures that if vertex $S$ of $\Upsilon_p^{i}$ has
%	a tag $v$, then no process has decided $1-v$ in $S(I^i)$.

%---------------------------------------------------------------
\subsection{Stabilization}
\label{subsec:cht:stabilization}
%---------------------------------------------------------------

Note that the simulation trees can only grow with time.
As a result, once a vertex of the simulation forest $\Upsilon_p$ gets a tag $v$, it
	cannot lose it later.
Thus, eventually every vertex of $\Upsilon_p$ stabilizes being
	$0$-valent, $1$-valent, or bivalent.
Since correct processes keep continuously exchanging the {\fd} samples
	and updating their simulation forests,
	every simulation tree computed by a correct process at any given time
	will eventually be a subtree of the simulation forest of every correct process.

Formally, let $p$ be any correct process, $t$ be any time, $i$
	be any index in $\{0,1,\ldots,n\}$, and $S$ be any vertex of
	$\Upsilon_p^{i}(t)$. % (here $\Upsilon_p^i(t)$ denotes the value of $\Upsilon_p^i$ at time $t$).
Then:% there is $[p_2,d_2,k_2]\to[p_1,d_3,k_3]$,a time $t'\geq t$ such that for all $t''\geq t'$:
\begin{enumerate}

\item[(i)] %For any vertex $S$ of  $\Upsilon_p^{i}(t)$, if $S$
	There exists a non-empty $V\subseteq \{0,1\}$ such that
	there is a time after which the valence of $S$ is $V$.
	(We say that the valence of $S$ \emph{stabilizes} on $V$ at $p$.)

\item[(ii)] If the valence of $S$ stabilizes on $V$ at $p$, then
	for every correct process $q$, there is a time after
	which $S$ is a vertex of $\Upsilon_q^{i}$ and
	the valence of $S$ stabilizes on $V$ at $q$.

\end{enumerate}

Hence, the correct processes eventually agree on the same tagged simulation
	subtrees.
In discussing the stabilized tagged simulation forest,
	it is thus convenient to consider the \emph{limit}
	infinite DAG $G$ and the \emph{limit} infinite simulation
	forest $\Upsilon=\{\Upsilon^0,\Upsilon^1,\ldots,\Upsilon^n\}$
	such that for all $i\in\{0,1,\ldots,n\}$ and all correct processes $p$,
	$\cup_{t\in\Time} G_p(t) = G$ and
	$\cup_{t\in\Time} \Upsilon_p^i(t) = \Upsilon^i$.
%It is important to see that $G$ contains paths $g$ such that

%---------------------------------------------------------------
\subsection{Critical index}
\label{subsec:cht:critical}
%---------------------------------------------------------------

Let $p$ be any correct process.
We say that index $i\in\{1,2,\ldots,n\}$ is \emph{critical}
	if \emph{either} the root of $\Upsilon^{i}$ is bivalent \emph{or}
	the root of $\Upsilon^{i-1}$ is $0$-valent and the root of $\Upsilon^{i}$ is $1$-valent.
In the first case, we say that $i$ is \emph{bivalent critical}.
In the second case, we say that
	$i$ is \emph{univalent critical}.

\begin{lemma}
\label{lemma:cht-critical}
There is at least one critical index in $\{1,2,\ldots,n\}$.
\end{lemma}
\begin{proof}
Indeed, by the Validity property of consensus, the root of $\Upsilon^{0}$ is $0$-valent,
	and the root of $\Upsilon^{1}$ is $1$-valent.
Thus, there must be an index $i\in\{1,2,\ldots,n\}$ such that
	the root of $\Upsilon^{i-1}$ is $0$-valent,
	and $\Upsilon^{i}$ is either
	$1$-valent or bivalent.
\end{proof}

\noindent
%Eventually, every correct process computes a \emph{stable} critical index.
Since tagged simulation forests computed at the correct processes tend
	to the same infinite tagged simulation forest,
	eventually, all correct processes compute the same \emph{smallest} critical
	index $i$ of the same type (univalent or bivalent).
Now we have two cases to consider for the smallest critical index: (1) $i$ is univalent critical, or
	(2) $i$ is bivalent critical.

\ignore{
Before extracting a correct process from the simulated forest,
	we observe the following:
\begin{lemma}
\label{lemma:cht-similar}
Let $S_0$ and $S_1$ be two vertices of $\Upsilon^i$ and $p$ be a process
	$p$ such that for all $q \in \Pi-\{p\}$,
	the states of $q$ %and the sets messages addressed to $q$
	in $S_0(I^i)$ and $S_1(I^i)$ are identical.
If $S_0$ is $0$-valent and $S_1$ is $1$-valent, then $p$ is correct.
\end{lemma}
\begin{proof}
By contradiction, assume that $p_i$ is faulty.
Let $S_0$ be induced by a path $g_1$ in $G$, and let $S_1$ be
	induced by a path $g_2$ in $G$.
Then $G$ contains infinite paths $g_1\cdot g$ and $g_2\cdot g$ such that in $g$,
	$p_i$ does not participate and every correct process participates infinitely often.
Then $\Upsilon^i$ contains a schedule $S$
	that is applicable to $S_0(I^i)$ (induced by a prefix of $g$ and $I^i$) in which
	$p_i$ does not take steps and every correct process $p$ decides.
Since $S_0$ is $0$-valent, $p$ decides $0$ in $S(I^i)$.
But $p_i$ is the only process that has different states in
	$I^{i-1}$, and $I^i$ and $p_i$ does not take part in $S$.
Thus, $S$ (induced by a prefix of $g$ and $I^{i-1}$) is also a vertex of $\Upsilon^{i-1}$,
	and $p$ decides $1$ in $S(I^{i-1})$.
But the root of $\Upsilon^{i-1}$ is $0$-valent
	--- a contradiction.
\end{proof}
}

%---------------------------------------------------------------
\subsubsection*{(1) Handling univalent critical index}
\label{subsec:cht:univalent}
%---------------------------------------------------------------

\begin{lemma}
\label{lemma:cht-univalent}
If $i$ is univalent critical, then $p_i$ is correct.
\end{lemma}
\begin{proof}
By contradiction, assume that $p_i$ is faulty.
Then $G$ contains an infinite path $g$ in which $p_i$ does not participate
	and every correct process participates infinitely often.
Then $\Upsilon^i$ contains a vertex $S$ %(induced by a prefix of $g$ and $I^i$)
	such that $p_i$ does not take steps in $S$ and some correct process $p$ decides in $S(I^i)$.
Since $i$ is $1$-valent, $p$ decides $1$ in $S(I^i)$.
But $p_i$ is the only process that has different states in
	$I^{i-1}$ and $I^i$, and $p_i$ does not take part in $S$.
Thus, $S$ %(induced by a prefix of $g$ and $I^{i-1}$)
	is also a vertex of $\Upsilon^{i-1}$
	and $p$ decides $1$ in $S(I^{i-1})$.
But the root of $\Upsilon^{i-1}$ is $0$-valent
	--- a contradiction.
\end{proof}

%---------------------------------------------------------------
\subsubsection*{(2) Handling bivalent critical index}
\label{subsec:cht:bivalent}
%---------------------------------------------------------------

Assume now that the root of $\Upsilon^i$ is \emph{bivalent}.
Below we show that $\Upsilon^i$ then contains a \emph{decision gadget}, i.e.,
	a finite subtree which is either a \emph{fork} or a \emph{hook}
	(Figure~\ref{fig:CHT-fork-hook}).

\begin{figure}[htbp]
  \centering
  \includegraphics[scale=.85]{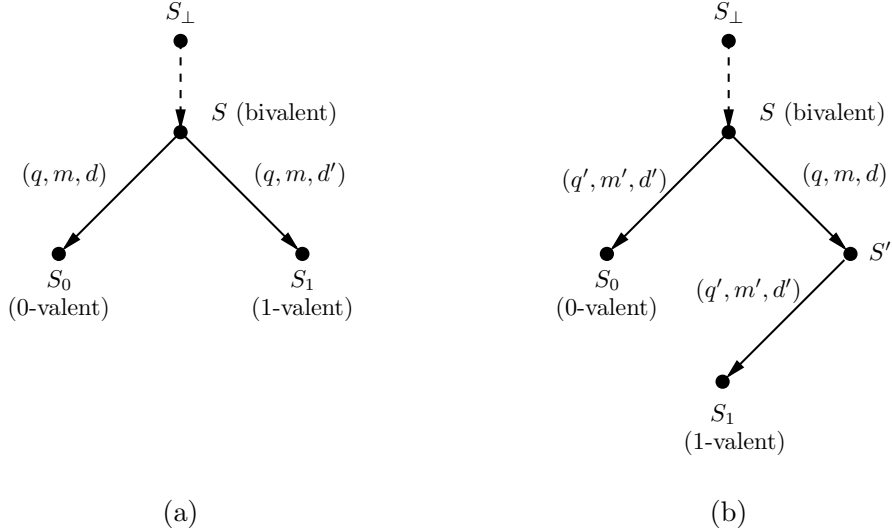}
  \caption{A fork and a hook}
  \label{fig:CHT-fork-hook}
\end{figure}

A fork (case (a) in Figure~\ref{fig:CHT-fork-hook}) consists of a bivalent
	vertex $S$ from which two \emph{different} steps by the
	\emph{same} process $q$, consuming the same message $m$,
	are possible which lead, on the one hand, to
	a $0$-valent vertex $S_0$ and, on the other hand, to a $1$-valent vertex
	$S_1$.

A hook (case (b) in Figure~\ref{fig:CHT-fork-hook}) consists of a
	bivalent vertex $S$, a vertex $S'$ which is reached by
	executing a step of some process $q$, and two vertices
	$S_0$ and $S_1$ reached by applying \emph{the same} step of process $q'$
	to, respectively, $S$ and $S'$.
Additionally, $S_0$ must be $0$-valent and $S_1$ must be $1$-valent (or vice
	versa; the order does not matter here).

In both cases, we say that $q$ is the \emph{deciding process}, and
	$S$ is the \emph{pivot} of the decision gadget.

\begin{lemma}
\label{lemma:cht-deciding}
The deciding process of a decision gadget is correct.
\end{lemma}
\begin{proof}
Consider any decision gadget $\gamma$ defined by a pivot $S$,
	vertices $S_0$ and $S_1$ of opposite
	valence and a deciding process $q$.
By contradiction, assume that $q$ is faulty.
Let $g$, $g_0$ and $g_1$ be the simulation stimuli of, respectively,
	$S$, $S_0$ and $S_1$.
Then $G$ contains an infinite path $\tilde g$ such that
	(a) $g\cdot \tilde g$, $g_0\cdot \tilde g$, $g_1\cdot \tilde g$
	are paths in $G$, and
	(b) $q$ does not appear
	and the correct processes appear infinitely often in $g$.

Let $\gamma$ be a fork (case (a) in Figure~\ref{fig:CHT-fork-hook}).
Then there is a finite schedule $\tilde S$ compatible with a prefix of $\tilde g$ and
	applicable to $S(I^i)$ such that some correct process
	$p$ decides in $S\cdot \tilde S(I^i)$;
	without loss of generality, assume that $p$ decides $0$.
Since $q$ is the only process that can distinguish $S(I^i)$ and $S_1(I^i)$,
	and $q$ does not appear in $\tilde S$,
	$\tilde S$ is also applicable to $S_1(I^i)$.
Since $g_1\cdot \tilde g$ is a path of $G$ and $\tilde S$ is compatible with a prefix of $\tilde g$,
	it follows that  $S_1\cdot \tilde S$ is a vertex of $\Upsilon^i$.
Hence, $p$ also decides $0$ in $S_1\cdot \tilde S(I^i)$.
But $S_1$ is $1$-valent --- a contradiction.

Let $\gamma$ be a hook (case (b) in Figure~\ref{fig:CHT-fork-hook}).
Then there is a finite schedule $\tilde S$ compatible with a prefix of $g$ and
	applicable to $S_0(I^i)$ such that some correct process
	$p$ decides in $S_0\cdot \tilde S(I^i)$.
Without loss of generality, assume that $S_0$ is $0$-valent, and hence $p$ decides $0$
	in  $S_0\cdot \tilde S(I^i)$.
Since $q$ is the only process that can distinguish $S_0(I^i)$ and $S_1(I^i)$,
	and $q$ does not appear in $\tilde S$,
	$\tilde S$ is also applicable to $S_1(I^i)$.
Since $g_1\cdot \tilde g$ is a path of $G$ and $\tilde S$ is compatible with a prefix of $\tilde g$,
	it follows that  $S_1\cdot \tilde S$ is a vertex of $\Upsilon^i$.
Hence, $p$ also decides $0$ in $S_1\cdot \tilde S(I^i)$
But $S_1$ is $1$-valent --- a contradiction.
\end{proof}

\noindent
Now we need to show that any bivalent simulation tree $\Upsilon^i$ contains at
	least one decision gadget $\gamma$.
%Then every correct process $p$ can eventually compute the deciding process $q$ of
%	$\gamma$. By Lemma~\ref{lemma:cht-deciding}, $q$ is correct.

\begin{lemma}
\label{lemma:cht-bivalent}
If $i$ is bivalent critical, then $\Upsilon^i$ contains a decision gadget.
\end{lemma}
\begin{proof}
%Take the \emph{infinite} bivalent simulation tree $\Upsilon^i$.
Let $i$ be a bivalent critical index.
In Figure~\ref{fig:CHT-gadget}, we present a procedure
	which goes through $\Upsilon^i$.
The algorithm starts from the bivalent root of $\Upsilon^i$ and
	terminates when a hook or a fork has been found.

\begin{figure}[htbp]
\hrule \vspace{2mm} {\small
\begin{tabbing}
 bbb\=bb\=bb\=bb\=bb\=bb\=bb\=bb \=  \kill
    \> $S\get S_{\bot}$\\
    \> \res{while} \id{true} \res{do}\\
    \>\> $p\get\langle$choose the next correct process in a round robin fashion$\rangle$\\
    \>\> $m\get\langle$choose the oldest undelivered message addressed to $p$ in $S(I^i)\rangle$\\
    \>\> \res{if} $\langle S$ has a descendant $S'$ in $\Upsilon^i$ (possibly $S=S'$) such that,
		for some $d$,\\
    \>\>\>\> $S'\cdot(p,m,d)$ is a bivalent vertex of $\Upsilon^i \rangle$ \\
    \>\> \res{then} $S\get S'\cdot (p,m,d)$\\
    \>\> \res{else} \res{exit}
 \end{tabbing}
\hrule }
  \caption{Locating a decision gadget}
  \label{fig:CHT-gadget}
\end{figure}

We show that the algorithm indeed terminates.
Suppose not.
Then the algorithm locates an infinite \emph{fair path} through the
	simulation tree, i.e., a path in which all correct processes get
	scheduled infinitely often and every message sent to a correct
	process is eventually consumed.
Additionally, this fair path goes through bivalent states only.
But no correct process can decide in a bivalent state $S(I^i)$
	(otherwise we would violate the Agreement property of consensus).
As a result, we constructed a run of $\A$ in which no correct process ever decides
	--- a contradiction.

Thus, the algorithm in Figure~\ref{fig:CHT-gadget} terminates.
That is, there exist a bivalent vertex $S$, a correct process $p$,
	and a message $m$ addressed to $p$ in $S(I^i)$ such that
	
\begin{description}
\item[(*)] For all descendants $S'$ of $S$ (including $S'=S$)
	and all $d$, $S'\cdot (p,m,d)$ is \emph{not}
	a bivalent vertex of $\Upsilon^i$.
%	(here $m$ is the oldest message addressed to $p$
%	in $S(I^i)$, or the null message if no such message exists).
\end{description}

In other words, any step of $p$ consuming message $m$ brings any descendant of $S$
	(including $S$ itself) to either a $1$-valent or a $0$-valent state.
Without loss of generality, assume that, for some $d$, $S\cdot(p,m,d)$ is
	a $0$-valent vertex of $\Upsilon^i$.
Since $S$ is bivalent, it must have a $1$-valent descendant $S''$.

If $S''$ includes a step in which $p$ consumes $m$, then we define $S'$
	as the vertex of $\Upsilon^i$ such that, for some $d'$, $S'\cdot (p,m,d')$ is a prefix of $S''$.
If $S''$ includes no step in which $p$ consumes $m$, then we define $S'=S''$.
Since $p$ is correct, for some $d'$, $S'\cdot (p,m,d')$ is a vertex of $\Upsilon^i$.
In both cases, we obtain $S'$ such that for some $d'$,
	$S'\cdot(p,m,d')$ is a $1$-valent vertex of $\Upsilon^i$.

Let the path from $S$ to $S'$ go through the vertices
	$\sigma_0=S,\sigma_1,\ldots,\sigma_{m-1},\sigma_m=S'$.
By transitivity of $G$, for all $k\in\{0,1,\ldots, m\}$,
	$\sigma_k\cdot (p,m,d')$ is a vertex of $\Upsilon^i$.
By (*), $\sigma_k\cdot(p,m,d')$ is either $0$-valent or $1$-valent
	vertex of $\Upsilon^i$.

\begin{figure}[htbp]
  \centering
  \includegraphics[scale=.87]{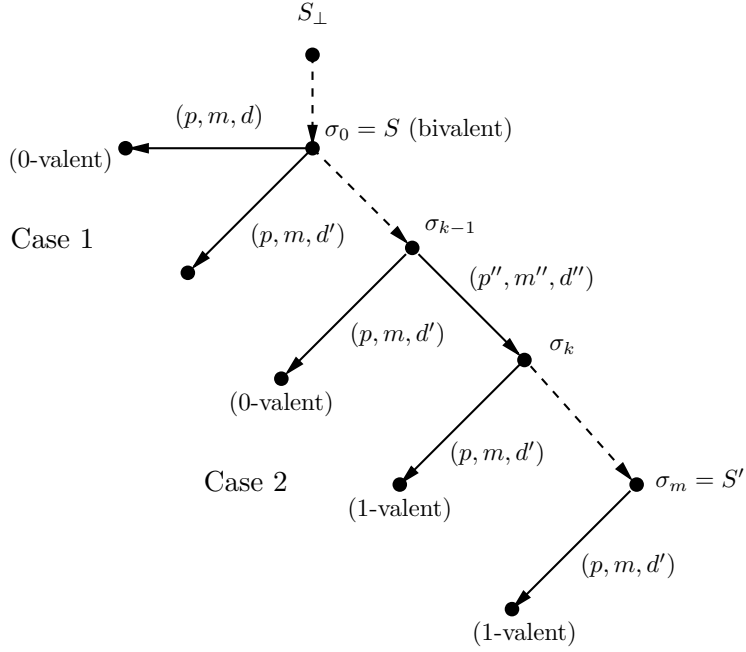}
  \caption{Locating a fork (Case~1) or a hook (Case~2)}
  \label{fig:CHT-path}
\end{figure}

Let $k\in\{0,\ldots,m\}$ be the lowest index such that $(p,m,d')$ brings $\sigma_k$
to a $1$-valent state. We know that such an index exists, since  $\sigma_m\cdot(p,m,d')$ is $1$-valent
and all such resulting states are either $0$-valent or $1$-valent.

Now we have the following two cases to consider: (1) $k=0$, and (2) $k>0$.

Assume that $k=0$, i.e., $(p,m,d')$ applied to $S$ brings it to a $1$-valent state.
But we know that there is a step $(p,m,d)$ that brings $S$ to a $0$-valent state
	(Case~1 in Figure~\ref{fig:CHT-path}).
That is, a fork is located!

If $k>0$, we have the following situation.
Step $(p,m,d')$ brings $\sigma_{k-1}$ to a $0$-valent state,
	and $\sigma_{k}=\sigma_{k-1}\cdot (p',m',d'')$ to a $1$-valent state
	(Case~2 in Figure~\ref{fig:CHT-path}).
But that is a hook!
				
As a result, any bivalent infinite simulation tree has at least one decision gadget.
\end{proof}

%---------------------------------------------------------------
\subsection{The reduction algorithm}
\label{subsec:cht:reduction}
%---------------------------------------------------------------

Now we are ready to complete the description of $T_{\D\to\Omega}$.
In the computation task (Figure~\ref{fig:CHT-comp}), every process $p$
	periodically extracts the current \emph{leader} from its
	simulation forest, so that eventually
	the correct processes agree on the same correct leader.
The current leader is stored in variable $\OO_p$.

\begin{figure}[htbp]
\hrule \vspace{2mm} {\small
\setcounter{linenumber}{0}
\begin{tabbing}
 bbb\=bb\=bb\=bb\=bb\=bb\=bb\=bb \=  \kill
\>Initially:\\
\>\> for $i=0,1,\ldots,n$: $\Upsilon_p^i \get $ empty graph\\
\>\> $\OO_p \get p$\\
\\
\> \textbf{while} true \textbf{do}\\
\\
\>\>\{ Build and tag the simulation forest induced by $G_p$ \}\\
\>\>	\textbf{for} $i=0,1,\ldots, n$	\textbf{do} \\
\>\>\> 		$\Upsilon_p^i \get$ simulation tree induced by $G_p$ and $I^i$\\
\>\>\>		\textbf{for} every vertex $S$ of $\Upsilon_p^i$:\\
\>\>\>\> 			\textbf{if} $S$ has a descendant $S'$ such that $p$ decides $v$
				in $S'(I^i)$ \textbf{then} \\
\>\>\>\>\>			add tag $v$ to $S$\\
\\
\>\>\{ Select a process from the tagged simulation forest \}\\
\>\>	\textbf{if} there is a critical index \textbf{then}\\
\>\>\> 		$i \get$ the smallest critical index\\
\>\>\>		\textbf{if} $i$ is univalent critical \textbf{then}  $\OO_p \get p_i$\\
\>\>\> 		\textbf{if} $\Upsilon_p^i$ has a decision gadget \textbf{then}\\
\>\>\>\>			$\OO_p \get$ the deciding process of the smallest
				decision gadget in  $\Upsilon_p^i$
\end{tabbing}
\hrule }
\caption{Extracting a correct leader: code for each process $p$}
\label{fig:CHT-comp}
\end{figure}

Initially, $p$ elects itself as a leader.
Periodically, $p$ updates its simulation forest $\Upsilon_p$ by incorporating more
	simulation stimuli from $G_p$.
If the forest has a univalent critical index $i$, then $p$ outputs $p_i$ as
	the current leader estimate.
If the forest has a bivalent critical index $i$ and $\Upsilon_p^i$ contains
	a decision gadget, then $p$ outputs the deciding process of \emph{the smallest} decision
	gadget in $\Upsilon_p^i$ (the ``smallest'' can be well-defined,
	since the vertices of the simulation tree are countable).

Eventually, the correct processes locate the same \emph{stable} critical index $i$.
Now we have two cases to consider:

\begin{enumerate}
\item[(i)] $i$ is univalent critical.
	By Lemma~\ref{lemma:cht-univalent}, $p_i$ is correct.
\item[(ii)] $i$ is bivalent critical.
	By Lemma~\ref{lemma:cht-bivalent}, the limit simulation tree $\Upsilon^i$
		 contains a decision gadget.
	Eventually, the correct processes locate the same decision gadget
		$\gamma$ in $\Upsilon_i$ and compute the deciding process $q$
		of $\gamma$.
	By Lemma~\ref{lemma:cht-deciding}, $q$ is correct.
\end{enumerate}

Thus, eventually, the correct processes elect the same correct leader
	--- $\Omega$ is emulated!

\end{document}

%% file: macros.tex
\newcommand{\btb}{\begin{tabbing}}
\newcommand{\etb}{\end{tabbing}}

\newcommand{\be}{\begin{itemize}}
\newcommand{\ee}{\end{itemize}}
\newcommand{\bn}{\begin{enumerate}}
\newcommand{\en}{\end{enumerate}}

\newcommand{\bc}{\begin{center}}
\newcommand{\ec}{\end{center}}

\newcommand{\bfg}{\begin{figure}[th]}
\newcommand{\efg}{\end{figure}}

\newcommand{\tpl}[1]{\ensuremath{\langle #1 \rangle}}

\renewcommand{\int}{\textit{int}}

\def\ioisize{\footnotesize}

\newcommand{\ioi}[3]{\bgroup\ioisize \begin{tabbing}
XX\=XX\=    \kill           %%Signature:
Input:\+ \\                 %%Input, Output, Internal
  #1 \-  \\
Output: \+ \\
  #2 \- \\
Internal: \+\\
  #3 \-
\end{tabbing} \egroup} 

\newcommand{\io}[2]{\bgroup\ioisize \begin{tabbing}
XX\=XX\=    \kill           %%Signature:
Input:\+ \\                 %%Input, Output
  #1 \-  \\
Output: \+ \\
  #2 \- 
\end{tabbing} \egroup} 

\newcommand{\ioti}[4]{\bgroup\ioisize \begin{tabbing}
XX\=XX\=    \kill           %%Signature:
Input:\+ \\                 %%Input, Output, Time-passing, Internal
  #1 \-  \\
Output: \+ \\
  #2 \- \\
Time-passing: \+ \\
  #3 \- \\
Internal: \+\\
  #4 \-
\end{tabbing} \egroup} 

\newcommand{\ei}[2]{\bgroup\ioisize \begin{tabbing}
XX\=XX\=    \kill           %%Signature:
External:\+ \\              %%External, Internal
  #1 \-  \\
Internal: \+\\
  #2 \-
\end{tabbing} \egroup}

\def\iocodesize{\scriptsize}
 
\newcommand{\iocodeonecol}[1]{\begin{center}\iocodesize
\begin{minipage}[t]{.95\linewidth}
\hbox to\linewidth{}
#1
\end{minipage}
\end{center}}

\newcommand{\iocode}[2]{\vspace{-15pt}\begin{center}\iocodesize
\begin{minipage}[t]{.475\linewidth}
\hbox to\linewidth{}
#1
\end{minipage}
\hspace{.01\linewidth} 
\begin{minipage}[t]{.475\linewidth}
\hbox to\linewidth{}
#2
\end{minipage}
\end{center}}

\newlength\Mwidth
\newlength\jsawidth

\def\statedef#1{\def\auxstatedef{#1}%
\def\auxempty{}%
\settowidth{\Mwidth}{M}
\jsawidth=\textwidth
\advance\jsawidth by -7cm
\advance\jsawidth by -\Mwidth
\advance\jsawidth by -10\tabcolsep
\advance\jsawidth by -6\arrayrulewidth
\begin{trivlist}\item[]%
\begin{tabular}{|p{2.5cm}|c|p{2.2cm}|p{2.3cm}|p{\jsawidth}|}
\hline
{\bf Variable} & \phantom{M} & {\bf Type} & {\bf Initially} & {\bf Description}\\
\hline\hline}
\def\endstatedef{\\ 
\ifx\auxstatedef\auxempty
  \hline
\else%
  \hline\hline
   \multicolumn{5}{|l|}{\auxstatedef}\\\hline
 \fi
\end{tabular}%
\end{trivlist}}

\def\statetext{\begin{minipage}[t]{\linewidth}}
\def\endstatetext{\end{minipage}\vspace{2pt}}

\usepackage{calc}       % To use the ``-'' below.

\renewcommand{\iocode}[3][.475]{\vspace{-15pt}\begin{center}\iocodesize
\begin{minipage}[t]{#1\linewidth}
\hbox to\linewidth{}
#2
\end{minipage}
\hspace{.01\linewidth} 
\begin{minipage}[t]{.95\linewidth-#1\linewidth}
\hbox to\linewidth{}
#3
\end{minipage}
\end{center}}

\newcommand{\MP}[1]{}
\newcommand{\rmv}[1]{}

\newcommand{\no}{\textrm{No}}

\newcommand{\OO}{\Omega\textit{-output}}

\newcounter{linenumber}

%% file: ecc-tr.bbl
\begin{thebibliography}{10}

\bibitem{Bre00}
E.~A. Brewer.
\newblock Towards robust distributed systems (abstract).
\newblock In {\em Proceedings of the Nineteenth Annual ACM Symposium on
  Principles of Distributed Computing}, PODC '00, pages 7--, 2000.

\bibitem{CHT96}
T.~D. Chandra, V.~Hadzilacos, and S.~Toueg.
\newblock The weakest failure detector for solving consensus.
\newblock {\em J.~ACM}, 43(4):685--722, July 1996.

\bibitem{CT96}
T.~D. Chandra and S.~Toueg.
\newblock Unreliable failure detectors for reliable distributed systems.
\newblock {\em J.~ACM}, 43(2):225--267, Mar. 1996.

\bibitem{Chang2008}
F.~Chang, J.~Dean, S.~Ghemawat, W.~C. Hsieh, D.~A. Wallach, M.~Burrows,
  T.~Chandra, A.~Fikes, and R.~E. Gruber.
\newblock Bigtable: A distributed storage system for structured data.
\newblock {\em ACM Trans. Comput. Syst.}, 26(2):4:1--4:26, June 2008.

\bibitem{CT94}
B.~Charron-Bost and G.~Tel.
\newblock Approximation d'une borne inf\'erieure r\'epartie.
\newblock Technical Report LIX/RR/94/06, Laboratoire d'Informatique LIX,
  \'Ecole Polytechnique, Sept. 1994.

\bibitem{Cooper2008}
B.~F. Cooper, R.~Ramakrishnan, U.~Srivastava, A.~Silberstein, P.~Bohannon,
  H.-A. Jacobsen, N.~Puz, D.~Weaver, and R.~Yerneni.
\newblock Pnuts: Yahoo!'s hosted data serving platform.
\newblock {\em Proc. VLDB Endow.}, 1(2):1277--1288, Aug. 2008.

\bibitem{DeCandia2007}
G.~DeCandia, D.~Hastorun, M.~Jampani, G.~Kakulapati, A.~Lakshman, A.~Pilchin,
  S.~Sivasubramanian, P.~Vosshall, and W.~Vogels.
\newblock Dynamo: Amazon's highly available key-value store.
\newblock In {\em Proceedings of Twenty-first ACM SIGOPS Symposium on Operating
  Systems Principles}, SOSP '07, pages 205--220, New York, NY, USA, 2007. ACM.

\bibitem{DFG10}
C.~Delporte-Gallet, H.~Fauconnier, and R.~Guerraoui.
\newblock Tight failure detection bounds on atomic object implementations.
\newblock {\em J.~ACM}, 57(4), 2010.

\bibitem{Dolev2010}
S.~Dolev, R.~I. Kat, and E.~M. Schiller.
\newblock When consensus meets self-stabilization.
\newblock {\em Journal of Computer and System Sciences}, 76(8):884 -- 900,
  2010.

\bibitem{Fekete1996}
A.~Fekete, D.~Gupta, V.~Luchangco, N.~Lynch, and A.~Shvartsman.
\newblock Eventually-serializable data services.
\newblock In {\em Proceedings of the Fifteenth Annual ACM Symposium on
  Principles of Distributed Computing}, PODC '96, pages 300--309, New York, NY,
  USA, 1996. ACM.

\bibitem{FLP85}
M.~J. Fischer, N.~A. Lynch, and M.~S. Paterson.
\newblock Impossibility of distributed consensus with one faulty process.
\newblock {\em J.~ACM}, 32(2):374--382, Apr. 1985.

\bibitem{fd-survey}
F.~C. Freiling, R.~Guerraoui, and P.~Kuznetsov.
\newblock The failure detector abstraction.
\newblock {\em ACM Comput. Surv.}, 43(2):9:1--9:40, Feb. 2011.

\bibitem{GL02}
S.~Gilbert and N.~Lynch.
\newblock Brewer's conjecture and the feasibility of consistent, available,
  partition-tolerant web services.
\newblock {\em SIGACT News}, 33(2):51--59, June 2002.

\bibitem{GHKT12}
R.~Guerraoui, V.~Hadzilacos, P.~Kuznetsov, and S.~Toueg.
\newblock The weakest failure detectors to solve quittable consensus and
  nonblocking atomic commit.
\newblock {\em SIAM J. Comput.}, 41(6):1343--1379, 2012.

\bibitem{GR14}
R.~Guerraoui and E.~Ruppert.
\newblock A paradox of eventual linearizability in shared memory.
\newblock In {\em Proceedings of the 2014 ACM Symposium on Principles of
  Distributed Computing}, PODC '14, pages 40--49, 2014.

\bibitem{HT94}
V.~Hadzilacos and S.~Toueg.
\newblock A modular approach to fault-tolerant broadcasts and related problems.
\newblock Technical Report TR 94-1425, Department of Computer Science, Cornell
  University, May 1994.

\bibitem{JT08}
P.~Jayanti and S.~Toueg.
\newblock Every problem has a weakest failure detector.
\newblock In {\em PODC}, pages 75--84, 2008.

\bibitem{KMO11}
F.~Kuhn, Y.~Moses, and R.~Oshman.
\newblock Coordinated consensus in dynamic networks.
\newblock In {\em Proceedings of the 30th Annual ACM Symposium on Principles of
  Distributed Computing (PODC)}, pages 1--10. ACM, 2011.

\bibitem{Cassandra}
A.~Lakshman and P.~Malik.
\newblock Cassandra: A decentralized structured storage system.
\newblock {\em SIGOPS Oper. Syst. Rev.}, 44(2):35--40, Apr. 2010.

\bibitem{Lam77}
L.~Lamport.
\newblock Proving the correctness of multiprocessor programs.
\newblock {\em Transactions on software engineering}, 3(2):125--143, Mar. 1977.

\bibitem{Lam98}
L.~Lamport.
\newblock The {Part-Time} parliament.
\newblock {\em ACM Transactions on Computer Systems}, 16(2):133--169, May 1998.

\bibitem{Lam06}
L.~Lamport.
\newblock Lower bounds for asynchronous consensus.
\newblock {\em Distributed Computing}, 19(2):104--125, 2006.

\bibitem{MRT00-mv}
A.~Mostefaoui, M.~Raynal, and F.~Tronel.
\newblock From binary consensus to multivalued consensus in asynchronous
  message-passing systems.
\newblock {\em Inf. Process. Lett.}, 73(5-6):207--212, Mar. 2000.

\bibitem{SS05}
Y.~Saito and M.~Shapiro.
\newblock Optimistic replication.
\newblock {\em ACM Comput. Surv.}, 37(1):42--81, Mar. 2005.

\bibitem{Sch90}
F.~B. Schneider.
\newblock Implementing fault-tolerant services using the state machine
  approach: A tutorial.
\newblock {\em ACM Computing Surveys}, 22(4):299--319, Dec. 1990.

\bibitem{Serafini2010}
M.~Serafini, D.~Dobre, M.~Majuntke, P.~Bokor, and N.~Suri.
\newblock Eventually linearizable shared objects.
\newblock In A.~W. Richa and R.~Guerraoui, editors, {\em Proceedings of the
  29th Annual ACM Symposium on Principles of Distributed Computing}, pages
  95--104. ACM, 2010.

\bibitem{Singh2009}
A.~Singh, P.~Fonseca, P.~Kuznetsov, R.~Rodrigues, and P.~Maniatis.
\newblock Zeno: Eventually consistent byzantine-fault tolerance.
\newblock In {\em Proceedings of the 6th USENIX Symposium on Networked Systems
  Design and Implementation}, NSDI'09, pages 169--184, Berkeley, CA, USA, 2009.
  USENIX Association.

\bibitem{Vogels2009}
W.~Vogels.
\newblock Eventually consistent.
\newblock {\em Commun. ACM}, 52(1):40--44, Jan. 2009.

\end{thebibliography}
